\documentclass[acmsmall]{acmart}

\setcopyright{acmlicensed}
\acmJournal{POMACS}
\acmYear{2024} \acmVolume{8} \acmNumber{1} \acmArticle{4} \acmMonth{3}\acmDOI{10.1145/3639030}

\usepackage{amsmath}
\usepackage{amsthm}
\usepackage{algorithm}
\usepackage{algorithmic}
\usepackage{xspace}
\usepackage{graphicx}
\usepackage{subfigure}
\usepackage{wrapfig}
\usepackage{graphicx}
\usepackage{subfigure}
\usepackage{url}
\usepackage{color}
\usepackage{multirow}
\usepackage{gensymb}
\usepackage{longtable}
\usepackage{tikz}
\usepackage{comment}
\usepackage{dsfont}
\usepackage{framed}
\usepackage{enumitem}
\usepackage{xspace}
\usepackage{bm}
\usepackage{pifont}

\theoremstyle{plain}
\newtheorem{theorem}{Theorem}[section]

\newtheorem{lemma}[theorem]{Lemma}

\theoremstyle{definition}
\newtheorem{definition}{Definition}
\newtheorem{assumption}{Assumption}

\newcommand{\ouralg}{\texttt{LA-OACP}\xspace}
\newcommand{\expert}{\texttt{OACP}\xspace}
\newcommand{\experttwo}{\texttt{OACP+}\xspace}
\newcommand{\ouralgtwo}{\texttt{OACP+}\xspace}

\AtBeginDocument{%
  }



\begin{document}

\title{Online Allocation with Replenishable Budgets: Worst Case and Beyond}

\author{Jianyi Yang}
\email{jyang239@ucr.edu}
\affiliation{%
	\institution{University of California, Riverside}
	\streetaddress{900 University Ave.}
	\city{Riverside}
	\state{California}
	\postcode{92521}
    \country{United States}
}

\author{Pengfei Li}
\email{pli081@ucr.edu}
\affiliation{%
  \institution{University of California, Riverside}
  \streetaddress{900 University Ave.}
  \city{Riverside}
  \state{California}
  \postcode{92521}
  \country{United States}
}

\author{Mohammad J. Islam}
\email{misla056@ucr.edu}
\affiliation{%
	\institution{University of California, Riverside}
	\streetaddress{900 University Ave.}
	\city{Riverside}
	\state{California}
	\postcode{92521}
    \country{United States}
}

\author{Shaolei Ren}
\email{sren@ece.ucr.edu}
\affiliation{%
	\institution{University of California, Riverside}
	\streetaddress{900 University Ave.}
	\city{Riverside}
	\state{California}
	\postcode{92521}
    \country{United States}
}

\renewcommand{\shortauthors}{Jianyi Yang, et al.}

\begin{abstract}
This paper studies online resource allocation with replenishable budgets,
where budgets can be replenished on top of the initial budget and
an agent sequentially chooses 
online allocation decisions without violating the available budget constraint
at each round. We propose a
novel online algorithm, called 
\expert (Opportunistic Allocation with Conservative Pricing), that 
conservatively adjusts dual variables while opportunistically
utilizing available resources.
\expert achieves a bounded asymptotic
competitive ratio in adversarial settings as the number of decision rounds $T$ gets large. Importantly, the asymptotic competitive ratio of \expert
is optimal in the absence of additional assumptions on budget replenishment.
To further improve the competitive ratio, we 
make a mild assumption that there is budget replenishment every
$T^*\geq1$ decision rounds and propose
\experttwo to dynamically adjust the total budget assignment for online allocation. 
Next,
we move beyond the worst-case and propose
\ouralg (Learning-Augmented \expert/\experttwo),
a novel learning-augmented algorithm
for online allocation with replenishable budgets.
We prove that \ouralg can improve
the average utility compared to \expert/\experttwo when the ML predictor
is properly trained, while still offering worst-case utility guarantees when the
ML predictions are arbitrarily wrong.
Finally, we run simulation studies of sustainable AI inference 
powered by renewables, validating our analysis
and demonstrating the empirical benefits of \ouralg.
\end{abstract}

\maketitle

\section{Introduction}

Online allocation subject to resource (or budget) constraints 
models a sequential decision-making problem where the agent
needs to allocate resources without violating the available budget constraint
at each round.
It is a central problem of critical importance
in numerous applications, such as revenue management, online advertising,
computing resource management, among many others. 
For example, 
Internet companies need to select advertisements based on online user arrivals
subject to advertisers' budget constraints;
cloud operators need to dynamically allocate user requests
to available machines subject to resource constraints;
and edge devices need to dynamically optimize its battery energy usage while intermittently
harvesting energy from the surrounding environment.
As such, the problem of online allocation and its variants
have received rich attention in the past few decades \cite{OnlineAllocation_Dual_mirror_descent_Google_ICML_2020_balseiro2020dual,OnlineAllocation_MultiInventory_MinghuaChen_Sigmetrics_2022_10.1145/3530902,OnlineAllocation_DualMirroDescent_Google_OperationalResearch_2022_doi:10.1287/opre.2021.2242,Shaolei_L2O_LAAU_OnlineBudget_AAAI_2023,NRM_book_talluri2004theory,resolving_bounded_revenue_loss_NRM_jasin2012re}.

Online allocation decisions are temporally coupled due to total budget constraints,
thus requiring complete offline information to obtain the optimal solution.
Nonetheless, the availability of only online information in practice
makes online allocation extremely challenging.
To meet budget constraints
in online settings, a commonly considered approach is 
 Lagrangian relaxation, which includes weighted budget constraints
 as a regularizer for online decision making where the weights
 are dual variables and can be interpreted as the budget/resource price \cite{devanur2019near,OCO_zinkevich2003online,Neely_Booklet,PalomarChiang_Distributed}.
Consequently,
by adjusting the resource price,
the agent's budget consumption is also governed so as to meet the budget constraint. For example,
 there have been a variety of
approaches to updating the dual variables online \cite{OCO_agrawal2014fast,OCO_zinkevich2003online,Neely_Booklet,OnlineAllocation_Dual_mirror_descent_Google_ICML_2020_balseiro2020dual,OnlineAllocation_DualMirroDescent_Google_OperationalResearch_2022_doi:10.1287/opre.2021.2242}.

Despite these efforts and advances in various (relaxed) settings
such as stochastic utility functions \cite{OnlineAllocation_Dual_mirror_descent_Google_ICML_2020_balseiro2020dual}, 
optimizing the total utility subject
to strict budget constraints still remains a challenging problem
in \emph{adversarial} settings, where the utility functions can be arbitrarily
presented to the agent. In fact, 
 competitive online algorithms for adversarial settings have only been proposed very recently. More concretely, 
 online resource allocation with a single-inventory constraint
\cite{OnlineAllocation_SingleInventory_MinghuaChen_Sigmetrics_2019_10.1145/3322205.3311081} and a multi-inventory constraint \cite{OnlineAllocation_MultiInventory_MinghuaChen_Sigmetrics_2022_10.1145/3530902} are two of the very few known competitive online algorithms with a finite number of decision rounds under the assumption that the utility 
 functions of each inventory are separable.
 In \cite{OnlineAllocation_DualMirroDescent_Google_OperationalResearch_2022_doi:10.1287/opre.2021.2242}, an online allocation algorithm that adjusts
 the dual variable is proposed, achieving a bounded asymptotic competitive ratio in adversarial settings when the length of each problem instance is sufficiently long.
Nonetheless, these studies \cite{OnlineAllocation_DualMirroDescent_Google_OperationalResearch_2022_doi:10.1287/opre.2021.2242,OnlineAllocation_MultiInventory_MinghuaChen_Sigmetrics_2022_10.1145/3530902,OnlineAllocation_SingleInventory_MinghuaChen_Sigmetrics_2019_10.1145/3322205.3311081} are crucially limited in the following  aspects.

$\bullet$ \emph{No budget replenishment}. First and foremost, the total budget constraint is fixed without allowing \emph{replenishment} online \cite{OnlineAllocation_DualMirroDescent_Google_OperationalResearch_2022_doi:10.1287/opre.2021.2242,OnlineAllocation_MultiInventory_MinghuaChen_Sigmetrics_2022_10.1145/3530902,OnlineAllocation_SingleInventory_MinghuaChen_Sigmetrics_2019_10.1145/3322205.3311081}. In fact, these algorithms explicitly assume
that budgets are \emph{not} replenishable, which would otherwise
void their competitive analysis.
However, budget replenishment in an online manner is common in practice, e.g., dynamic energy harvesting (see Section~\ref{sec:formulation_example}
for additional examples). While
some studies \cite{EnergyHarvesting_LearningAided_LyapunovOptimization_LongboHuang_TMC_2020_8807278,learning_aided_energy_harvesting_outdated_state_infomation_yu2019learning,Lyapunov_optimization_energy_harvesting_wireless_sensor_qiu2018lyapunov,OnlineOpt_EnergyHarvesting_Bregman_Neely_USC_SIgmetrics_2020_10.1145/3428337,Huang:2011:UOS:2107502.2107531} have considered budget replenishment,
they typically focus on independent and identically distributed  budget replenishment. In contrast, arbitrary budget replenishment
in adversarial settings naturally provides additional power to the adversary, thus
creating significantly more challenges.

$\bullet$ \emph{Worst-case performance only.}
Second, the studies \cite{OnlineAllocation_DualMirroDescent_Google_OperationalResearch_2022_doi:10.1287/opre.2021.2242,OnlineAllocation_MultiInventory_MinghuaChen_Sigmetrics_2022_10.1145/3530902,OnlineAllocation_SingleInventory_MinghuaChen_Sigmetrics_2019_10.1145/3322205.3311081} only focus on the worst-case performance
in terms of the competitive ratio. As a result, the
conservativeness needed to  address
the worst possible problem input significantly limits their 
average-case performance for most typical problem inputs.
Online algorithms based on machine learning (ML)
models have been considered for various problems \cite{L2O_NewDog_OldTrick_Google_ICLR_2019,L2O_Combinatorial_Reinforcement_AAAI_2020,L2O_OnlineBipartiteMatching_Toronto_ArXiv_2021_DBLP:journals/corr/abs-2109-10380,Shaolei_L2O_LAAU_OnlineBudget_AAAI_2023}, including
online resource allocation \cite{L2O_OnlineResource_PriceCloud_ChuanWu_AAAI_2019_10.1609/aaai.v33i01.33017570,L2O_AdversarialOnlineResource_ChuanWu_HKU_TOMPECS_2021_10.1145/3494526}.
Nonetheless, unlike the hand-crafted online algorithms \cite{OnlineAllocation_DualMirroDescent_Google_OperationalResearch_2022_doi:10.1287/opre.2021.2242,OnlineAllocation_MultiInventory_MinghuaChen_Sigmetrics_2022_10.1145/3530902,OnlineAllocation_SingleInventory_MinghuaChen_Sigmetrics_2019_10.1145/3322205.3311081},
ML-based online optimizers may not offer worst-case performance
guarantees and can result in significantly bad results when, for example,
the training-testing distribution differs. Even though heuristic techniques
such adversarial training can empirically mitigate the lack of performance robustness to some extent,
it is still challenging to provably guarantee the worst-case performance of ML models. Thus, it remains an open problem to achieve the \emph{best of both worlds} --- improving the average utility while offering the worst-case robustness (in the presence of budget replenishment). 
In fact, as highlighted above, there even do not exist competitive online algorithms that address budget replenishment in adversarial settings, let alone
a learning-augmented algorithm that can improve the average performance while provably offering  worst-case performance guarantees.

\textbf{Contributions.}
In this paper, we address the above points and consider online allocation with replenishable budgets, where
the agent receives budget replenishment on the fly and needs to choose irrevocable online decisions to allocate $M$ resources. The goal of the agent
is to maximize the total utility over $T$ rounds subject to 
per-round available budget constraints, where the per-round utility is a function in terms of the online allocation decision.

We first consider an adversarial setting and propose an online algorithm, called
\expert (Opportunistic Allocation with Conservative Pricing), that 
updates the dual variable (i.e., resource pricing) online to regulate the
agent's budget allocation and achieves 
an asymptotic competitive ratio as $T\to\infty$.
The key insight of \expert is
that we treat the uncertain budget replenishment differently
than the initially-assigned fixed budget and set the resource price
in a conservative manner, which encourages the agent to be more frugal
while still allowing the agent
to opportunistically utilize the replenished budgets when applicable.
Most importantly, we prove in Theorem~\ref{thm:expertbound}
that \expert achieves the same asymptotic competitive ratio bound
as the state-of-the-art optimal bound
in \cite{OnlineAllocation_DualMirroDescent_Google_OperationalResearch_2022_doi:10.1287/opre.2021.2242} that does not address budget replenishment.
In our setting with replenishable budgets, the adversary naturally has
more power than the setting of a fixed known budget, as it can arbitrarily
present budget replenishments to the agent.
Therefore, achieving the same asymptotic competitive ratio 
as that of the state-of-the-art algorithm for fixed budget allocation \cite{OnlineAllocation_DualMirroDescent_Google_OperationalResearch_2022_doi:10.1287/opre.2021.2242} highlights the benefit of \expert in terms of addressing additional uncertainties of replenished budget.

\begin{table*}[t]
\scriptsize
\centering
\begin{tabular}{ l|c|c|c|c} 
\toprule
 \textbf{Algorithm} &  \textbf{Budget replenishment} & 
  \textbf{Budget cap} &  \textbf{Worst-case robustness} & 
   \textbf{Average utility bound}\\
\midrule
\text{CR-Pursuit} \cite{OnlineAllocation_SingleInventory_MinghuaChen_Sigmetrics_2019_10.1145/3322205.3311081}&\ding{56}&NA&\ding{52}&\ding{56}\\ 
\midrule
\text{A\&P \cite{OnlineAllocation_MultiInventory_MinghuaChen_Sigmetrics_2022_10.1145/3530902}} &\ding{56}&NA&\ding{52}&\ding{56}\\ 
\midrule
\text{DMD} \cite{OnlineAllocation_DualMirroDescent_Google_OperationalResearch_2022_doi:10.1287/opre.2021.2242} &\ding{56}&NA&\ding{52}&\ding{56}\\ 
\midrule
\textbf{\expert} (our work) &\ding{52}&\ding{52}&\ding{52}&\ding{56}\\ 
\midrule
\textbf{\experttwo} (our work) &\ding{52}&\ding{52}&\ding{52}&\ding{56}\\ 
\midrule
\textbf{\ouralg} (our work) &\ding{52}&\ding{52}&\ding{52}&\ding{52}\\ 
\bottomrule
\end{tabular}
\caption{Comparison between our work and recent online competitive allocation algorithms for adversarial settings. Algorithms for non-adversarial settings
are discussed in Section~\ref{sec:related} and not shown in the table.}
\label{table:comparison_literature}
\end{table*}

Next, we propose \experttwo to utilize the budget replenishment more efficiently under a mild assumption that the budget is replenished at least every $T^*\geq1$ decision rounds.   
Specifically, 
\experttwo divides the whole episode of $T$ rounds into 
$K$ frames of unequal lengths and performs frame-level budget assignment online and a round-level online budget allocation within each frame.  
To account for the maximum budget cap, a new threshold-based budget assignment strategy is proposed to decide the assigned budget for each frame.  Given the assigned budget for each frame, we apply \expert for round-level budget allocation while deferring all the budget replenishment to future frames.  
We prove that
\experttwo achieves a higher asymptotic competitive ratio than \expert if the total budget replenishment is positive
in every $T^*$ rounds (Theorem~\ref{thm:enhanced_expertbound_cr}). 

Last but not least, we move beyond the worst-case and
aim to maximize the average utility while still offering
worst-case utility guarantees.
We propose a novel learning-augmented algorithm, called \ouralg
(Learning-Augmented \expert),
that integrates a trained ML predictor with \expert.
More concretely, \ouralg utilizes the ML prediction (i.e.,
online allocation decision by the ML-based optimizer)
and expert decision (from \expert or \experttwo) as advice, and judiciously 
combine them. The key novelty of \ouralg is to introduce
a new reservation utility that produces a constrained
decision set within which all decisions can meet the worst-case utility
constraint (defined with respect to \expert or \experttwo).
Meanwhile, \ouralg ensures that the online decisions are chosen
from the constrained decision set while being close
to ML predictions so as to exploit the benefits of ML predictions
to improve the average utility.
We rigorously prove that \ouralg can improve
the average utility compared to \expert when the ML predictor
is properly trained, while still offering worst-case utility guarantees
(see Theorems~\ref{thm:robustness} and~\ref{thm:average}).

Finally, we run simulation studies of sustainable AI inference
to maximize
the total utility 
subject to energy constraints with renewable replenishment. Our results  validate the analysis of \expert, \experttwo and \ouralg, demonstrating the empirical advantage of \ouralg in terms of the average utility
over \expert and \experttwo as well as other baseline algorithms.

We highlight the main difference between our algorithms
and recent online allocation algorithms that consider adversarial
settings  in Table~\ref{table:comparison_literature}.
\emph{Our major contributions are also summarized as follows}.
First, we propose two novel online algorithms \expert and \experttwo
that achieve bounded asymptotic competitive ratios
for online allocation with replenishable budgets in adversarial settings
(\textbf{Theorem~\ref{thm:expertbound}  and Theorem~\ref{thm:enhanced_expertbound_cr}}).
To our knowledge, the proposed provably-competitive algorithms advance
the existing competitive online algorithms to address budget replenishment in adversarial settings for the first time \cite{OnlineAllocation_MultiInventory_MinghuaChen_Sigmetrics_2022_10.1145/3530902,OnlineAllocation_DualMirroDescent_Google_OperationalResearch_2022_doi:10.1287/opre.2021.2242,OnlineAllocation_SingleInventory_MinghuaChen_Sigmetrics_2019_10.1145/3322205.3311081}.
Second, we move beyond the worst case
and propose a novel learning-augmented algorithm, \ouralg,
that probably improves the average utility compared to \expert or \experttwo (\textbf{Theorem~\ref{thm:average}}), while still offering worst-case
utility guarantees for online allocation with budget replenishment for any problem instance (\textbf{Theorem~\ref{thm:robustness}}).

\section{Problem Formulation}

In this section, we present the problem formulation for online
allocation with replenishable budgets.

\textbf{Notations:} For the convenience of presentation, we first introduce the common notations used throughout
the paper. Unless otherwise noted, we use $[N]$ to denote the set $\{1,2,\cdots, N\}$ for a positive integer $N$. 
 $\mathbb{E}(\cdot)$ is the expectation operator,  $\mathbb{P}$ is a probability measure, 
 $\mathbb{I}(x)$ is an indicator function
 (i.e., $\mathbb{I}(x)=1$ if the condition $x$ is true and
 $\mathbb{I}(x)=0$ otherwise),
 and $\mathcal{R}_+^D$ and $\mathcal{R}_{++}^D$ are $D$-dimensional non-negative and strictly positive real number spaces, respectively.
 For a vector $x$, 
$x_j$ denotes its $j$-the element and $\|x\|$ is its norm  ({$l_2$ norm by default). For two vectors $x$ and $y$,
we use $x\leq y$ to denote \emph{element}-wise inequality, i.e.,
$x_j\leq y_j$ for all $j$ and use $x\odot y$ to denote \emph{element}-wise product. $\min(x,y)$ denotes the \emph{element}-wise minimization. We also use $[x]^b=\min\left(x,b\right)$ and
$[x]^+=\max\left(x,0\right)$, where the capping and rectifying operators
are applied for each element when $x$ is a vector.
For a sequence of 
variables $c_1,\cdots, c_T$, we use $c_{i:j}$
to denote the subsequence $c_i,\cdots,c_j$ for $1\leq i\leq j\leq T$;
we have $c_{i:j}=\varnothing$ if $i>j$.

\subsection{Model}

We consider an online allocation problem with replenishable resource budgets, where each sequence (a.k.a., problem instance)
includes $T$ consecutive rounds and involves sequential
allocation of $M$ types of resources based on online information.

At the beginning of a sequence (i.e., round $t=1$), 
the decision maker (i.e., agent)
is endowed with an initial resource budget $B_1=[B_{1,1},\cdots, B_{M,1}]\in\mathcal{R}_{++}^M$,
where $B_{m,1}=T\rho_m$ 
is the initial resource budget
for type-$m$ resource, with $\rho_m>0$ being
 the per-round average budget initially assigned to the 
 agent, 
for $m\in[M]$. Moreover, we have $B_1\leq B_{\max}$, where
$B_{\max}=[B_{1,\max},\cdots, B_{M,\max}]$
represents the maximum budget cap.
The inclusion
 of $B_{\max}$ is both practical and general:
$B_{\max}$ captures practical constraints such
as battery capacity for energy resources,
space constraint for product inventory, among others,
and the budget cap can be effectively voided when setting a large $B_{m,\max}\to\infty$ for $m\in[M]$, to which case our design also applies.

At the beginning of each round $t\in[T]$,
the agent is presented with a utility function $f_t(x): \mathcal{R}_{+}^M
\to \mathcal{R}_{+}$, where $x\in\mathcal{X}$ is the allocation decision.
 Additionally, the agent also receives a potential budget replenishment $\hat{E}_t=\left[\hat{E}_{1,t},\cdots,\hat{E}_{M,t}\right]\in\mathcal{R}_{+}^M$,
resulting in a total budget of $\min\left(B_t+\hat{E}_t, {B_{\max}}\right)= B_t + E_t$,
available for allocation
at round $t$, where $B_t$ is the remaining budget at the end of round $t$. 
In other words, due to the budget cap, the actual budget replenishment is $E_t=\min(\hat{E}_t, B_{\max}-B_t)$
 at round $t$. 

 The agent's allocation decision for round $t$ is $x_t=[x_{1,t},\cdots, x_{M,t}]\in\mathcal{X}$, where $\mathcal{X}=\{x\in\mathcal{R}_+^M| 0\leq x \leq \bar{x}\}$ with $\bar{x}=[\bar{x}_1,\cdots, \bar{x}_M]$ representing
the maximum allocation for each resource type at each round. 
Note that we have $\bar{x}\leq B_{\max}$ since otherwise
the maximum budget cap is more stringent
while the maximum allocation constraint $\bar{x}$ is never binding.

Given
the budget replenishment and the agent's allocation decision,
the budget evolves as $B_{t+1}=\min \left(B_t+\hat{E}_t,B_{\max}\right)-x_t
=B_t+E_t-x_t$
for round $t+1$.
Thus, the information revealed to the agent at the beginning
of each round $t$ can be summarized as $y_t=(f_t,\hat{E}_t)$,
while all the information for a sequence can be written
as $y=[y_1,\cdots,y_T]\in\mathcal{Y}$, where 
$\mathcal{Y}$ denotes the space of all possible episodic 
information. When the context is clear, we also use $y$
to denote a sequence. Any remaining budgets at the end of an sequence
are wasted without rolling over to the next sequence. If an algorithm $\pi$ is used to solve the problem with information $y$, the total the total utility is denoted as $F_T^{\pi}(y)=\sum_{t=1}^T f_t(x_t)$. 

To summarize, for a sequence
$y$, the \emph{offline} problem can be formulated
as 
\begin{subequations}\label{eqn:objective_offline}
 \begin{gather}\label{eqn:objective}
       \max_{x_{1:T}\in \mathcal{X}^T} \sum_{t=1}^Tf_t(x_{t}) \\
       \label{eqn:constraint_1}
       s.t., \;\;\;\; x_t\leq B_t+E_t  \text{ and } x_t\in\mathcal{X},
        \;\;\;\; \forall t\in[T]\\
        \label{eqn:budget_dynamics}
        B_{t+1}=B_t+{E}_t - x_t \text{ and } E_t=\min\left(\hat{E}_t, B_{\max} - B_t\right), \;\;\;\; \forall t\in[T]
 \end{gather}
\end{subequations}

 Next, we make the following standard assumptions on the utility function $f_t(x)$ for $t\in[T]$.

\begin{assumption}[Utility function $f_t(x)$]\label{assumption:utility_function}
For any $t\in [T]$, the utility function $f_t(x): \mathcal{R}_{+}^M
\to \mathcal{R}_{+}$  is assumed to be non-negative, 
 have subgradients at each point of $x\in\mathcal{X}$.
 In addition, we assume
 $f_t(0)=0$ and $\sup f_t(x)=\bar{f}$ for $t\in[T]$ and $x\in\mathcal{X}$.
\end{assumption}

The assumptions are standard in the literature on online allocation
with budget constraints \cite{OnlineAllocation_Dual_mirror_descent_Google_ICML_2020_balseiro2020dual,OnlineAllocation_DualMirroDescent_Google_OperationalResearch_2022_doi:10.1287/opre.2021.2242,OnlineAllocation_MultiInventory_MinghuaChen_Sigmetrics_2022_10.1145/3530902}. 
Note that we do not require concavity of the utility functions, making our algorithms applicable for a wide range of applications.

\subsection{Performance Metrics}

With complete information $y=[y_1,\cdots, y_T]\in\mathcal{Y}$
provided to the agent at the beginning of a sequence,
the problem in \eqref{eqn:objective_offline} can be efficiently solved
via subgradient methods for constrained optimization \cite{subgradient_goffin1977convergence,constrained_opt_bertsekas2014constrained,subgradient_boyd2003subgradient}. If the utility functions are concave,  subgradient methods such as the projected subgradient method and the primal-dual subgradient method have provable convergence guarantees \cite{subgradient_boyd2003subgradient,subgradient_goffin1977convergence}.  Nonetheless,
in practice, the agent only has access to online information
$y_{1:t}$ before making its decision $x_t$ at round $t\in[T]$, 
adding substantial challenges.

Our goal is to design an online algorithm $\pi$ that maps
available online information $y_{1:t}$ to a decision 
 $x_t\in \mathcal{X}$ subject to the budget constraint \eqref{eqn:constraint_1}
at each round $t\in[T]$. To measure the decision quality
of an online algorithm $\pi$, we use the following metrics that
capture the \emph{worst}-case and \emph{average}-case performance, respectively.

\begin{definition}[Asymptotic competitive ratio \cite{balseiro2019learning,borodin2005online}]\label{def:cr}
The asymptotic competitive ratio
of an online algorithm $\pi$ is $CR^{\pi}$ if $\lim_{T\to\infty}\sup_{y\in\mathcal{Y}}\frac{1}{T}\left(OPT(y)-\frac{1}{CR^{\pi}} F_T^{\pi}(y)\right)\leq 0$,
where
$F_T^{\pi}(y)=\sum_{t=1}^T f_t(x_t)$ is the total utility of algorithm $\pi$ 
and $OPT(y)$ is the optimal utility obtained by the oracle given
offline information.
\end{definition}

\begin{definition}[Average utility]
Given an online algorithm $\pi$, the average
utility is defined as $AVG^{\pi}=\mathbb{E}_{y\in\mathcal{Y}}\left[F_T^{\pi}(y)\right]$,
where the expectation is over the sequence information $y\sim \mathbb{P}_{y}$.
\end{definition}

Both competitive ratio and average utility are important in practice, characterizing
the robustness of an online algorithm (in terms of
its utility ratio to the optimal oracle) and its quality for typical problem instances, respectively. 
Here, we consider an asymptotic competitive ratio
(in the sense of $T\to\infty$)
because of the intrinsic hardness of our problem ---
even for online allocation of a fixed budget without
replenishment, only an asymptotic competitive ratio
is attainable in the state of the art \cite{OnlineAllocation_DualMirroDescent_Google_OperationalResearch_2022_doi:10.1287/opre.2021.2242}.
 We shall design in Section~\ref{sec:expert} online allocation algorithms to address the worst-case robustness, while
we will consider the average performance (subject to a worst-case robustness constraint) in Section~\ref{sec:learning}.

\subsection{Application Examples}\label{sec:formulation_example}

We now provide a few examples as motivating applications to make our model more concrete.

\emph{Online advertising with budget replenishment.} Online advertisement
serves as a prominent, if not the most prominent, source of revenue for Internet companies \cite{OnlineAllocation_DualMirroDescent_Google_OperationalResearch_2022_doi:10.1287/opre.2021.2242}. Advertisers need to dynamically set a biding
budget, which will then be used by
the publisher to maximize  profits or the number of impressions for advertisers per their contracts with the publisher. Meanwhile, they can also increase budgets anytime they like. Thus, by viewing
the bidding budget as an online decision, this problem fits nicely into
the online allocation of replenishable budgets.

\emph{Sustainable AI inference.} Nowadays, the rapidly increasing  demand for artificial intelligence (AI) inference, especially large language models, has resulted
in large carbon emissions \cite{ML_Carbon_Bloom_176B_Sasha_Luccioni_arXiv_2022_luccioni2022estimating}. To achieve sustainable AI inference, it is important to 
exploit renewable generation to replenish on-site energy storage.
Meanwhile, for the same AI inference service, there often
exist multiple models (e.g.,
eight different GPT-3 models \cite{ML_GPT3_Energy_Others_NIPS_2020_NEURIPS2020_1457c0d6}), each having a distinct model
size to offer a flexible tradeoff
between accuracy performance and energy consumption.
However, the renewables are known to be time-varying and unstable. Thus, 
by
viewing the intermittent renewables as replenished budgets,
 the resource manager needs to schedule an appropriate AI model
 for inference in an online manner to maximize the utility (e.g. maximizing the accuracy) given
available energy constraints \cite{carbon_aware_computing_radovanovic2022carbon,green_AI_schwartz2020green}.

\emph{Online inventory management with dynamic replenishment.}
Manufacturers need to dynamically dispatch available inventory to
different distributors based on market demands. Meanwhile, they will
also replenish the inventory through newly manufactured products.
The goal is to manage the available inventory to maximize the total
profit/revenue given dynamic replenishment and environment (e.g., market
demands and supply-chain situation), to which our model is well suited.

\section{\expert: Opportunistic Allocation with Conservative Pricing}\label{sec:expert}

In this section, we address the worst-case robustness in adversarial settings
and design
an asymptotically competitive online algorithm, called
\expert, that conservatively updates the dual variable based on mirror descent and opportunistically allocates replenished budgets.
 Using a novel technique,
\expert provably offers the optimal worst-case 
performance guarantees for adversarial settings of online allocation with replenishable budgets (Theorem~\ref{thm:expertbound}).
Then, by making
an additional assumption on the minimum budget replenishment,
we extend \expert to \experttwo, which offers an improved
asymptotic competitive ratio (Theorem~\ref{thm:enhanced_expertbound_cr}).

To solve the online allocation problem in \eqref{eqn:objective_offline}, one can equivalently relax the budget constraints
using Lagrangian
techniques. More specifically, 
instead of directly solving 
\eqref{eqn:objective_offline},
we introduce a regularizer and solve $\hat{x}_t=\arg\max_{x\in\mathcal{X}}\{f_t(x)-\mu_t^\top x\}$
where $\mu_t\in\mathcal{R}_+^M$ is the Lagrangian multiplier vector (a.k.a., \emph{dual} variable) with each entry corresponding
to one resource budget constraint. The interpretation
of $\mu_t$ is that it can be viewed
as the resource \emph{price} \cite{Varian_Microeconomics1992,PalomarChiang_Distributed}: a higher price encourages resource conservation to meet the budget constraints, and vice versa. 

If we were able to optimally set  $\mu_t\in\mathcal{R}_+^M$ for $t\in[T]$, 
we could optimally solve \eqref{eqn:objective_offline} 
while satisfying the per-round budget constraints.
Nonetheless, like in the original problem \eqref{eqn:objective_offline}, finding the optimal
$\mu_t$ for $t\in[T]$ requires the complete offline information $y=[y_1,\cdots,y_T]$ at the beginning
of an episode, but this information is clearly lacking 
for online allocation.

Despite this challenge, the interpretation of the dual variable $\mu_t$ as
the resource price at round $t\in[T]$ provides us with inspiration
for the design of \expert.
Specifically, in view of the dynamic budget replenishment
$E_t$, we propose to \emph{conservatively} update the
price $\mu_{t+1}$ to a higher value for each round ${t+1}$ 
as if $E_t$ does not exist, and then \emph{opportunistically}
use the actually available budget $B_t+E_t$.
Our algorithm, called \expert,
is described in Algorithm~\ref{alg:expert}.

\subsection{Competitive Algorithm Design}
\begin{algorithm}[!t]
\caption{Opportunistic Allocation with Conservative Pricing (\expert)}
	\begin{algorithmic}[1]\label{alg:expert}
  \REQUIRE Initialize dual variable $\mu_1$,
 and budget  {$B_1=\rho T$ for $\rho>0$}
  \FOR   {$t=1$ to $T$}
  \STATE Receive utility function $f_t(x)$ and potential budget replenishment $\hat{E}_t$.
        \STATE Get the actual replenished budget $E_t=\min\{\hat{E}_t,B_{\max}-B_t\}$ 
        \STATE Pre-select action $x_t$ based on $\mu_t$:
        $\hat{x}_t=\arg\max_{x\in \mathcal{X}}\{f_t(x)-\mu_t^\top x\}$\\
        \IF{$\hat{x}_t\leq B_t+E_t$}
        \STATE $x_t=\hat{x}_t$ and $g_t=-\hat{x}_t+\rho$ \;\;\; \textsf{//for conservative pricing}
        \ELSE
        \STATE $x_t=0$ and $g_t=0$
        \ENDIF
	\STATE Update budget $B_{t+1}=B_{t}+E_t -x_t$ 
  \STATE Dual mirror descent:\\
  $\mu_{t+1}=\arg\min_{\mu\geq0} g_t^\top\mu+\frac{1}{\eta}V_h(\mu,\mu_t)$
  where $V_h(\mu,\mu_t)=h(\mu)-h(\mu_t)-\triangledown h(\mu_t)^{\top}(\mu-\mu_t)$ is the Bregman divergence in which $h(\mu)$ is a 
  $\sigma$-strongly convex reference function (Assumption~\ref{assumption:reference_function})
  \ENDFOR
 \end{algorithmic}
\end{algorithm}

At each round
$t\in[T]$, given
$\mu_t$ and online information, we solve 
the following relaxed optimization problem:
\begin{equation}\label{eqn:relaxed_opt_each_round}
\hat{x}_t=\arg\max_{x\in \mathcal{X}}\{f_t(x)-\mu_t^\top x\}.
\end{equation}
Next, we check if $\hat{x}_t$ satisfies the current
budget constraint $B_t+E_t$:  we set $x_t=\hat{x}_t$ if the budget constraint
is satisfied, and $x_t=0$ otherwise. Then, we update
the dual variable 
based on mirror descent
 $\mu_{t+1}=\arg\min_{\mu\geq0} g_t^\top\mu+\frac{1}{\eta}V_h(\mu,\mu_t)$,
  where $V_h(\mu,\mu_t)=h(\mu)-h(\mu_t)-\triangledown h(\mu_t)^{\top}(\mu-\mu_t)$ is the Bregman divergence
defined with respect to a reference function $h(\mu)$.

The goal of mirror descent is to update the dual variable $\mu_{t+1}$ such that it can set a resource \emph{price} that reflects the current budget level   
while staying not too far away from the current
dual variable $\mu_t$ as regularized
by $\frac{1}{\eta}V_h(\mu,\mu_t)$ in
terms of
 Bregman divergence. 
 In particular, the usage of mirror descent to update dual variables for online constrained optimization 
 has begun to be explored recently
\cite{OnlineAllocation_Dual_mirror_descent_Google_ICML_2020_balseiro2020dual,OnlineAllocation_DualMirroDescent_Google_OperationalResearch_2022_doi:10.1287/opre.2021.2242,OnlineOpt_EnergyHarvesting_Bregman_Neely_USC_SIgmetrics_2020_10.1145/3428337}.
Nonetheless, the prior studies on online allocation under {adversarial}
settings 
have only considered a fixed budget without dynamic budget replenishment \cite{OnlineAllocation_DualMirroDescent_Google_OperationalResearch_2022_doi:10.1287/opre.2021.2242}.

\textbf{Key insight.} The key insight of \expert
lies in how we set $g_t$ and choose $x_t$ in Lines~5 and~6 of
Algorithm~\ref{alg:expert}. 
The dual variable $\mu_t$ is updated based on $g_t=-\hat{x}_t+{\rho}$,
whose inverse (i.e., $\hat{x}_t-{\rho}$) measures the overuse of the current allocation compared with a reference per-round budget ${\rho}=\frac{B_1}{T}$. When $g_t$ is smaller, the degree of budget overuse is greater, and $\mu_{t+1}$ tends to be greater in the mirror descent step, encouraging the agent to use fewer resources at round $t+1$. Under the setting of no budget replenishment, 
it is natural to set the per-round budget ${\rho}=\frac{B_1}{T}$
 to evaluate the degree of over-consumption for each round. Nonetheless, in the presence of budget replenishment, we cannot simply use $\rho+E_t$ as the reference to incorporate new replenishment $E_t$ in resource pricing. 
The reason is that the sequence of $E_t$ can be arbitrary and the future replenishment $E_{t+1},\cdots, E_T$ is unknown.
As a result, aggressively using $\tilde{\rho}=\rho+E_t$ as the reference per-round resource consumption can result in an unnecessarily low
resource price $\mu_{t+1}$.
Instead, \expert still sets the reference per-round budget as ${\rho}=\frac{B_1}{T}$ as if no budget replenishment were received.
Consequently, the resource price
$\mu_{t+1}$  tends to be higher than using $\rho+E_t$ otherwise, encouraging the agent to be more frugal in resource consumption.
 On the other hand, the budget replenishment $E_t$ can
be still used opportunistically by increasing the actual available budget
from $B_t$ to $B_t+E_t$ (Line~5). 
Thus, by doing so, \expert tends to be \emph{more conservative in resource pricing (i.e., $\mu_t$), while still opportunistically using budget replenishment in actual allocation decisions}.

Next, to make Algorithm~\ref{alg:expert} self-contained, we specify the following assumptions on the reference function
$h(\mu)$ used in the mirror descent step.

\begin{assumption}[Reference function $h(\mu)$]\label{assumption:reference_function}
The reference function $h(\mu):\mathcal{R}_+^M\to\mathcal{R}$ is 
 differentiable
 and $\sigma$-strongly convex in 
 $\|\cdot\|_1$-norm in $\mathcal{R}_+^M$, i.e.,
 $h(\mu)-h(\mu')\geq \triangledown h(\mu')^{\top}(\mu-\mu')+\frac{\sigma}{2}\|\mu-\mu'\|_1^2$ for
 any $\mu,\mu'\in\mathcal{R}_+^M$. 
 \end{assumption}
 
Assumption~\ref{assumption:reference_function}
is standard in the analysis of mirror descent-based
algorithms \cite{OnlineAllocation_Dual_mirror_descent_Google_ICML_2020_balseiro2020dual,OnlineAllocation_DualMirroDescent_Google_OperationalResearch_2022_doi:10.1287/opre.2021.2242}.
Along with Assumption~\ref{assumption:utility_function}
on the utility function, it essentially ensures
that there is always a unique solution 
in the mirror descent step in Line~11 of Algorithm~\ref{alg:expert}. Importantly,
this step can recover common gradient-based update algorithms by a proper choice of the reference function. For example,  with
$h(\mu)=\sum_{m=1}^M \mu_m \log (\mu_m)$, the
update in Line~11 of Algorithm~\ref{alg:expert}
becomes $\mu_{t+1}=\mu_t\odot\exp(-\eta g_t)$ and captures
multiplicative weight updates, where
the operator ``$\odot$'' is the element-wise product \cite{Alg_MultiplicativeWeightsUpdate_Arora_2012_v008a006};
for $h(\mu)=\frac{1}{2}\|\mu\|_2^2$,
the update rule becomes $\mu_{t+1}=\left[\mu_t-\eta g_t\right]^+$
and recovers online sub-gradient descent method \cite{OnlineAllocation_DualMirroDescent_Google_OperationalResearch_2022_doi:10.1287/opre.2021.2242}.

\subsection{Performance Analysis}

We proceed with the analysis of \expert in terms of its worst-case performance. Our result highlights that \expert is asymptotically competitive
against the offline oracle $OPT$, generalizing
the prior results on the allocation of a fixed budget
\cite{OnlineAllocation_DualMirroDescent_Google_OperationalResearch_2022_doi:10.1287/opre.2021.2242} to replenishable budgets.

\begin{theorem}\label{thm:expertbound}
For any episode $y\in\mathcal{Y}$ and $\eta>0$, by Algorithm \ref{alg:expert}, the utility of \expert satisfies
\begin{equation}\label{eqn:expert_bound}
\begin{split}
OPT(y)-\alpha F_T^{\expert}(y)
\leq  
\alpha\bar{f}+
\frac{\alpha \left(\bar{\rho}+\|\bar{x}\|_{\infty}\right)^2\eta T}{2\sigma}+\frac{\alpha}{\eta}V_h(\mu,\mu_1),
\end{split}
\end{equation}
where $\alpha= \max_{m\in[M]} 
\frac{\bar{x}_m}{\rho_m}$, $\bar{\rho}=\max_{m\in[M]}\rho_m$ is the maximum per-round average budget initially assigned to the agent at round $t=1$,
$\bar{x}$ is the maximum per-round resource allocation constraint,
$V_h(\mu,\mu_1)$ is the Bregman divergence between
$\mu$ and the initial dual variable $\mu_1$ given
the $\sigma$-strongly convex reference function $h$,
and $\mu=0$ if Line~5 of Algorithm~\ref{alg:expert} is
always true, and otherwise, $\mu=\frac{\bar{f}}{\alpha\rho_j}e_j$ with $j=\arg\min_{m\in\mathcal{M}_A} V_h(\frac{\bar{f}}{\alpha\rho_m}e_m,\mu_1)$ where $\mathcal{M}_A=\left\{m|\; \exists t\in[T] \text{ such that }\hat{x}_{m,t}>(B_t+E_t)_m\right\}$, $e_m$
is a standard $M$-dimensional unit vector. Furthermore, by optimally setting $\eta=\frac{1}{\bar{\rho}+\|\bar{x}\|_{\infty}}\sqrt{2\sigma V_h(\mu,\mu_1)/T}$, we have
\begin{equation}\label{eqn:asymptotic_cr_sublinear}
\begin{split}
\lim_{T\to\infty}\sup_{y\in\mathcal{Y}}\frac{1}{T}\left(OPT(y)-\alpha F_T^{\expert}(y)\right)\leq \lim_{T\to\infty}\frac{1}{T}\left(\alpha\bar{f} + \alpha \left(\bar{\rho}+\|\bar{x}\|_{\infty}\right)\sqrt{\frac{V_h(\mu,\mu_1)T}{2\sigma}}\right)=0,
\end{split}
\end{equation}
i.e., \expert achieves an asymptotic
competitive ratio of $\frac{1}{\alpha}=\min_{m\in[M]}\frac{\rho_m}{\bar{x}_m}$ against $OPT$.\footnote{Throughout the paper, the asymptotic competitive ratio is naturally no greater than 1, i.e., $CR^{\expert}=\min\{1, \frac{1}{\alpha}\}$.}
\end{theorem}

The proof of Theorem~\ref{thm:expertbound} is deferred to  Appendix~\ref{appendix:proof_expert} to keep the main body
of the paper more concise for
better readability.
Our proof relies on a technique specifically designed for budget replenishment. 
Concretely, without budget replenishment, the allocation algorithm (e.g., DMD in \cite{OnlineAllocation_DualMirroDescent_Google_OperationalResearch_2022_doi:10.1287/opre.2021.2242}) stops allocation whenever any resource type in
the initial budget
$B_1$ is exhausted.
In contrast, \expert continues allocation
until the end of an episode due to new budget replenishment. To account for this, we introduce a group $\mathcal{T}_A$ of rounds that each has a budget violation event, and bound the total utility for rounds that are not in $\mathcal{T}_A$.

Theorem~\ref{thm:expertbound} can be interpreted as follows.
Without optimally setting $\eta$, by rearranging the terms in \eqref{eqn:expert_bound}, we have
$F_T^{\expert}(y)\geq \frac{1}{\alpha} OPT(y)
- \bar{f}- 
\frac{\left(\bar{\rho}+\|\bar{x}\|_{\infty}\right)^2\eta T}{2\sigma}-\frac{1}{\eta}V_h(\mu,\mu_1)$. That is,
 for any episode $y\in\mathcal{Y}$, \expert can obtain
a total utility of at least $\frac{1}{\alpha}$ times the optimal
oracle's utility, minus per-round utility bound $\bar{f}$ and a
term related to the convergence of $\mu$.
Moreover, by setting
$\eta\sim O(\frac{1}{\sqrt{T}})$, \expert achieves
an asymptotic competitive ratio bound of $\frac{1}{\alpha}$ as $T\to\infty$.
The parameter $\alpha= \max_{m\in[M]} 
\frac{\bar{x}_m}{\rho_m}$ measures how stringent
the
initially assigned per-round budget is with respect to the agent's own maximum
allocation constraint. Naturally, the larger $\alpha$ (i.e., 
the initial budget is relatively more limited), a lower competitive
ratio bound. Moreover, the asymptotic competitive ratio bound
in Theorem~\ref{thm:expertbound} 
matches the optimal bound for online allocation of a fixed budget
\cite{OnlineAllocation_DualMirroDescent_Google_OperationalResearch_2022_doi:10.1287/opre.2021.2242}.

We also note that, with the added uncertainties
due to  budget replenishment, the optimal (offline) resource price $\mu_t^*$
can also be time-varying, while the optimal resource price
$\mu^*$ is fixed when without budget replenishment \cite{OnlineAllocation_DualMirroDescent_Google_OperationalResearch_2022_doi:10.1287/opre.2021.2242}.
Consequently, even if we aggressively update
the resource price $\mu_{t}$ by directly incorporating replenished budgets
at each round,
there is still no hope to learn the optimal dynamic resource price $\mu_t^*$ with a sublinear regret (or an asymptotic competitive ratio of 1);
instead, we can incur additional utility losses due to aggressive but
potentially incorrect tracking of $\mu_t^*$ in an adversarial setting.
Therefore, \expert utilizes the design of conservative pricing
while using opportunistic allocation for actual decisions.
It adds to the literature by generalizing the state-of-the-art (asymptotically) competitive
online algorithm for the setting of a fixed budget \cite{OnlineAllocation_DualMirroDescent_Google_OperationalResearch_2022_doi:10.1287/opre.2021.2242}
to replenishable budgets.

In our setting with replenishable budgets, the adversary naturally has
more power than the setting of a fixed budget, as it can adversrially
present budget replenishments to the agent.
Thus, achieving the same optimal asymptotic competitive ratio 
as that of state-of-the-art DMD for fixed budget allocation \cite{OnlineAllocation_DualMirroDescent_Google_OperationalResearch_2022_doi:10.1287/opre.2021.2242} demonstrates the merit of \expert in terms of addressing additional uncertainties
of replenished budget.

Importantly,
our asymptotic competitive ratio $\frac{1}{\alpha}$ is optimal
in the adversarial budget replenishment setting.
Specifically, 
in the adversarial case,
it is possible that there is zero budget replenishment, or
the budget replenishment only arises in the last
decision round and the utility function for this round is chosen as zero by the adversary.
As a consequence, the replenished budget cannot be utilized to improve
the utility, and our setting essentially reduces to the no budget replenishment setting in the worst case.
This means that without further assumptions
on the budget replenishment, one cannot find a higher competitive ratio than the optimal bound $\frac{1}{\alpha}$ for online allocation with a fixed budget \cite{OnlineAllocation_DualMirroDescent_Google_OperationalResearch_2022_doi:10.1287/opre.2021.2242}.

\subsection{Extension to \experttwo with Minimum Budget Replenishment Assumption}\label{sec:experttwo}

\begin{algorithm}[!t]
\caption{Opportunistic Allocation with Conservative Pricing + (\experttwo)}
	\begin{algorithmic}[1]\label{alg:expert_2}
  \REQUIRE Unit frame length $T^*$
 and initial budget  {$B_1=\rho T$ for $\rho>0$}
 \FOR {frame $i=1$ to $K$}
 \STATE Initialize $\mu_{T_{i-1}+1}$, set learning
 rate $\eta_i>0$, assign the budget $B_{T_{i-1}+1}^{(i)}=B^{(i)}$ as Eqn.~\eqref{eqn:allocatebudget} and $\hat{\rho}_i=B^{(i)}/(T_i-T_{i-1})$,
 where $T_i=(2^{i}-1)T^*$. 
  \FOR   {$t=T_{i-1}+1$ to $T_i$}
  \STATE Receive utility function $f_t(x)$.
        \STATE Pre-select action $x_t$ based on $\mu_t$:
        $\hat{x}_t=\arg\max_{x\in \mathcal{X}}\{f_t(x)-\mu_t^\top x\}$\\
        \IF{$\hat{x}_t\leq B_t^{(i)}$}
        \STATE $x_t=\hat{x}_t$ and $g_t=-\hat{x}_t+\hat{\rho}_i$ \;\;\; 
        \ELSE
        \STATE $x_t=0$ and $g_t=0$
        \ENDIF
	\STATE Update budget $B_{t+1}^{(i)}=B_{t}^{(i)} -x_t$ and the actual remaining budget $B_{t+1}=B_{t}+E_t -x_t$
  \STATE Update dual
  $\mu_{t+1}=\arg\min_{\mu\geq0} g_t^\top\mu+\frac{1}{\eta_i}V_h(\mu,\mu_t)$.
  \ENDFOR
  \ENDFOR
 \end{algorithmic}
\end{algorithm}

In the unrestricted adversarial budget replenishment case, there can be zero
budget replenishment and hence, one cannot expect
a higher asymptotic competitive
ratio than that of the optimal bound for fixed budget allocation.
Next, to avoid the trivial case of no budget replenishment
and improve the asymptotic competitive
ratio, we make a mild assumption on the minimum budget replenishment every $T^*$
rounds (referred to as a unit frame)
and propose a new algorithm called \experttwo.


\subsubsection{The Design of \experttwo}\label{sec:design_experttwo}

As discussed in the key insight of Algorithm~\ref{alg:expert}, aggressively setting $g_t=-\hat{x}_t+\rho+E_t$ for resource pricing cannot improve the competitive ratio since $E_t$ is arbitrary and $\rho+E_t$ is not a reliable reference per-round budget in the adversarial case. On the other hand,
a higher fixed budget means
that the online allocator is less starved and hence
can increase the competitive ratio \cite{OnlineAllocation_Dual_mirror_descent_Google_ICML_2020_balseiro2020dual,OnlineAllocation_DualMirroDescent_Google_OperationalResearch_2022_doi:10.1287/opre.2021.2242}. Thus, this provides us with
an inspiration to improve the competitive ratio of \expert:
\emph{Batching the budget replenishment and allocating it later as if we had a higher fixed budget.}

Concretely, we design a new  two-level online allocation algorithm,
called \ouralgtwo, which divides an entire episode of $T$ rounds into $K$ frames and batches the budget replenishment in frame $i$ for resource allocation
in frame $i+1$. Then, 
within each frame, \ouralgtwo views
the effective budget replenishment (subject to frame-level budget allocation to be specified in Eqn.~\eqref{eqn:allocatebudget}) in the previous frame as
if it were a fixed resource and allocates it online.
 
 \experttwo is described in Algorithm~\ref{alg:expert_2},
 where we introduce a unit frame of length $T^*\geq1$ rounds during
 which a minimum amount of budget is replenished (see Definition~\ref{asp:replenishment}). Note that \ouralgtwo
 only needs the information of $T^*$, but does not know
 the minimum budget replenishment within $T^*$ rounds.
Within each frame $i\in [K]$ starting from round $T_{i-1}+1$ to round $T_i$, we initialize the dual variable, assign the budget $B^{(i)}$ as the initial budget for frame $i$, and set the reference per-round budget $\hat{\rho}_i=B^{(i)}/(T_i-T_{i-1})$. Then, by considering
that all the budget replenishment in frame $i$ is deferred
for allocation in frame $i+1$ (Line~11 of Algorithm~\ref{alg:expert_2}), 
we apply \expert with a fixed assigned frame-level budget $B^{(i)}$ to choose actions for all rounds in frame $i$. The dual variable is updated based on the reference per-round budget $\hat{\rho}_i$ and learning rate $\eta_i$ for frame $i$. 
Note that in Line~6, we make sure the allocation is not larger than the remaining frame budget $B_t^{(i)}$ which is a part of the fixed assigned frame-level budget $B^{(i)}$. This means that
the new budget replenishment in frame $i$ is not incorporated in the resource pricing or used for allocation in frame $i$. The remaining frame-level budget $B_t^{(i)}$ and the actual remaining budget $B_t$ are updated simultaneously in Line~11. 
By batching the budget replenishment in frame $i$ and
deferring it for allocation in frame $i+1$, 
\ouralgtwo 
can allocate more resources as if it had a higher fixed budget in frame $i+1$.

Nonetheless, to improve the competitive ratio, 
there are two key challenges in the design of \ouralgtwo
---
frame construction and frame-level budget assignment  --- which
we address as follows.

\textbf{Frame construction.} 
To defer the budget replenishment in one frame to
the next frame and allocate it as fixed budget,
it is crucial to appropriately decide the length of each frame, i.e., frame construction.
An intuitive way of frame construction is to divide the entire episode
 of $T$ rounds uniformly into $K=\left\lceil T/T^*\right\rceil$ frames, each with $T^*\geq1$ rounds (which is the length of a unit frame).  
 By doing so, \ouralgtwo incurs an additional term of
 $\mathcal{O}(\sqrt{T^*})$  in the reward bound of each frame
 by Theorem~\ref{thm:expertbound}
 and hence a total additional term of $\mathcal{O}\left(\sqrt{T^*}\left\lceil T/T^*\right\rceil\right)$, which grows linearly
 with $T$. 
Thus, to avoid the additional linear term $\mathcal{O}\left(\sqrt{T^*}\left\lceil T/T^*\right\rceil\right)$,
 \experttwo utilizes a doubling frame construction as follows.

Specifically, the entire episode of $T$ rounds is divided into $K=\left \lceil \log_2(T/T^*) \right \rceil $ frames, where $T^*\geq1$
is the length of a unit frame. The $i$-th frame starts from 
round $T_{i-1}+1$ and ends at round  $T_i$,
where $T_i=(2^{i}-1)T^*$.\footnote{The last frame
(i.e., $K$-th frame) starts from round $(2^{K-1}-1)T^*+1$
and ends at the last round $T$. For the convenience of presentation,
we assume $T=(2^K-1)T^*$ to be consistent with the previous frame's ending round $T_i=(2^{i}-1)T^*$.}
In other words, assuming the first frame has
a length of $T^*$ rounds,
the length of frame $i=2,\cdots,K$
is $2^{i-1} T^*$,
doubling the length
of its previous frame $i-1$. 
For each frame $i$, the additional term incurred by
\ouralgtwo is   $\mathcal{O}(\sqrt{2^{i-1} T^*})$,
the sum of which is still sublinear with respect to $T$,
keeping the asymptotic competitive ratio independent
of the choice of the initial dual in each frame.



\textbf{Frame-level budget assignment.} It remains to set the frame-level budget $B^{(i)}$ for each frame $i$ given uncertain future budget replenishment. The initial fixed budget $B_1=T\rho$
is proportionally divided into $K$ frames: 
the frame-level budget $B^{(i)}$ for each frame $i$ includes a fixed budget $2^{i-1}T^*\rho$, where $2^{i-1}T^*$ is the length of frame $i$. 
Additionally, the assigned frame-level budget $B^{(i)}$ also includes an additive budget $\Omega_i$ which comes from the budget replenishment and unused budgets assigned in previous frames. Without a maximum budget cap (i.e. $B_{\max}=\infty$), we can directly set $\Omega_i$ as the actual budget accumulation $B_{T_{i-1}+1}-(T-T_{i-1})\rho$, where $B_{T_{i-1}+1}$ is the actual remaining budget at the beginning of frame $i$ and $(T-T_{i-1})\rho=(T-(2^{i-1}-1)T^*)\rho$ is the sum of the fixed budget 
assignment reserved for the remaining frames (including frame $i$).
Thus, by combining the fixed budget and replenished budget (including unused assignments) from previous frames, the assigned total budget for frame $i$ is $B^{(i)}=B_{T_{i-1}+1}-(T-(2^i-1)T^*)\rho$.

However, if the maximum budget cap $B_{\max}$ exists, it can restrict the actual budget replenishment.
Thus, if we assign all the actually accumulated budget for frame $i$,
it can happen that 
little additional budgets (other
than the fixed budget $2^{i}T^*\rho$) can be used for frame $i+1$.
To further explain this point,
consider an online allocation problem with a linear utility function $f_t(x)=<c_t,x>$ (i.e., the inner product of $c_t$ and $x$). 
Suppose that the remaining budget $B_{T_1+1}$ at the beginning of the second frame (which has a length $2T^*$ rounds)
is as large as $B_{\max}$. This is possible if there is a large budget replenishment during the first frame.
As a result, new budget replenishments cannot be accumulated due to the budget cap $B_{\max}$ unless some budgets have been consumed.
Assume that the budget replenishment $\hat{E}_t$ and context parameter $c_t$ for the second frame are as follows. In the first $T^*+1$ rounds of the second frame, the budget replenishment is $\hat{E}_t>0$ and the context is $c_t=0$; in the following $T^*-1$ rounds of the second frame, the budget replenishment for each round is $\hat{E}_t=0$ and the context parameter $c_{t}$ is sufficiently large.
In this example, \ouralgtwo will not consume any resource during the first $T^*+1$ rounds, and instead consume all of the assigned budget $B^{(2)}$ during
the remaining $T^*-1$ rounds. 
As a result, no budget replenishment can be accumulated in this frame due to the maximum budget cap.
 If we still assign the frame budget as $B^{(2)}=B_{T^*+1}
 -(T-3T^*)\rho$ as if there were no budget cap,  the remaining budget at the beginning of the third frame will be $
 B_{T_2+1}=(T-3T^*)\rho$ and the assigned total budget for the third frame will be $4T^*\rho$, resulting in zero additive budget for the third frame ($\Omega_3=0$) other than the fixed budget assignment $4T^*\rho$.


To ensure a positive additive budget for each future frame, we need to allocate the budget replenishment $\Omega_i$ for frame $i$ more conservatively: the additive budget $\Omega_i, i\in[2,K-1]$ is set as the minimum of the actual budget accumulation $B_{T_{i-1}+1}-(T-(2^{i-1}-1)T^*)\rho$ and a threshold $\Gamma_i$, i.e., $\Omega_i=\min\{B_{T_{i-1}+1}-(T-(2^{i-1}-1)T^*)\rho,\Gamma_i\}$. 

It remains to design a proper threshold $\Gamma_i$ for frame-level budget assignment. Naturally,
if the budget cap $B_{\max}$ becomes larger, the threshold $\Gamma_i$ 
 should be set higher; also,  $\Gamma_i$ should increase with the length of the frame. We set the threshold as $\Gamma_i=2^{i-2}T^* \rho_{\max}\odot\beta$ where the operator ``$\odot$'' is the element-wise product, $\rho_{\max}=\frac{B_{\max}}{T}$ and $\beta\in R_+^M$ is a hyper-parameter indicating
 the level of conservativeness to balance between the aggressive budget assignment
 for the next frame and conservative budget reservation for subsequent future frames. Therefore, the assigned total frame-level budget for frame $i$ is the sum of the fixed budget assignment $2^{i-1}T^*\rho$ (where $\rho=\frac{B_1}{T}$) and an additive budget $\Omega_i$, i.e.
\begin{equation}\label{eqn:allocatebudget}
\begin{split}
B^{(i)}=2^{i-1}T^*\rho+\Omega_{i},
\end{split}
\end{equation}
where $\Omega_1=0$,
$\Omega_{i}=\min\left\{B_{T_{i-1}+1}-(T-(2^{i-1}-1)T^*)\rho, \Gamma_i\right\}$ with $\Gamma_i=2^{i-2}T^*\rho_{\max}\odot\beta$
for $i\in[2, K-1]$,
and $\Omega_{K}=B_{T_{K-1}+1}-2^{K-1}T^*\rho$.
When $\rho_{\max}=\frac{B_{\max}}{T}$ is sufficiently large
such that the threshold $\Gamma_i$ is not activated,
 the assigned budget becomes $B^{(i)}=B_{T_{i-1}+1}-(T-(2^{i}-1)T^*)\rho$,
 which reduces to the budget assignment without a maximum budget cap
 and shows the flexibility of our design of frame-level budget assignment.



\subsubsection{Performance Analysis}

In this section, we give the asymptotic competitive ratio of \experttwo to highlight the benefits of budget replenishment. To
avoid the adversarial case which can 
reduce to the no budget replenishment setting, 
we first define the minimum replenishment $E_{\min}\geq0$ for a unit frame with length $T^*$ and then provide the asymptotic competitive ratio relying on $E_{\min}$. The assumption of minimum budget replenishment in each unit frame is reasonably mild in practice, especially for large $T^*$. For example,
it is reasonable to assume that a minimum amount of solar renewables are replenished each day  \cite{carbon_aware_computing_radovanovic2022carbon,OnlineOpt_EnergyHarvesting_Bregman_Neely_USC_SIgmetrics_2020_10.1145/3428337,OnlineAllocation_DualMirroDescent_Google_OperationalResearch_2022_doi:10.1287/opre.2021.2242}. 
Note that $E_{\min}$ is 
decided by the environment and \ouralgtwo does not need
the knowledge of $E_{\min}$.

 \begin{definition}[Minimum budget replenishment]\label{asp:replenishment}
Given a unit frame of $T^*\geq 1$ rounds, the minimum potential budget replenishment for type-$m$ resource within each unit frame is $E_{\min,m}\geq 0$, i.e., 
$E_{\min,m} = \inf_j \left\{\sum_{t=(j-1)T^*+1}^{j\cdot T^*}\hat{E}_{t,m}\right\}$, where $\hat{E}_{t,m}$ is the budget that would be replenished at
round $t$ if $B_{\max,m}\to\infty$,
$j=0,\cdots,\left\lceil T/T^*\right\rceil-1$ is the index of a unit frame and $E_{\min}=\left[E_{\min,1}, \cdots, E_{\min,M}\right]$.
\end{definition}

\begin{theorem}\label{thm:enhanced_expertbound_cr}
If the learning rate for frame $i$ is chosen as $\eta_i=\frac{1}{\bar{\rho}+\frac{\bar{\beta}}{2}\bar{\rho}_{\max}+\|\bar{x}\|_{\infty}}\sqrt{2\sigma V_h(\mu,\mu_{T_{i-1}+1})/(2^{i-1}T^*)}$ with $\bar{\rho}_{\max}=\max_m\rho_{\max,m}$ where $\rho_{\max,m}=\frac{B_{\max,m}}{T}$ and $\bar{\beta}=\max_m\beta_{m}$, \experttwo achieves an asymptotic
competitive ratio against $OPT$ as 
\begin{equation}
CR^{\experttwo}=\min_{m\in[M]}\frac{\rho_m+\Delta\rho_m}{\bar{x}_m},
\end{equation}
where $\bar{x}_m$ is the maximum per-round allocation of type-$m$ resource and $\Delta\rho_m\geq 0$ is the improvement due to budget replenishment.
Specifically, if  $B_{\max,m}\geq (T+T^*)\rho_m$ holds for a resource $m$, we have $\Delta\rho_m = \min\left\{ \frac{E_{\min,m}}{2T^*}, \frac{2B_{\max,m}}{3(T+T^*)}-\frac{\rho_m}{3}\right\}$
with the optimal choice of $\beta_m=\frac{4T}{3(T+T^*)}-\frac{2\rho_m}{3\rho_{\max,m}}$;
and if $B_{\max,m}< (T+T^*)\rho_m$ holds for a resource $m$, we have $\Delta\rho_m = \min\left\{ \frac{E_{\min,m}}{2T^*}, \frac{B_{\max,m}}{6T^*}-\frac{(T-T^*)\rho_m}{6T^*}\right\}$ with
the optimal choice of $\beta_m=\frac{T}{3T^*}-\frac{T-T^*}{3T^*}\frac{\rho_m}{\rho_{\max,m}}$. 
Moreover, without the minimum budget replenishment
(i.e., $E_{\min,m}=0$), we have $\Delta\rho_m=0$ and the asymptotic competitive ratio $CR^{\experttwo}$ reduces to the one in Theorem~\ref{thm:expertbound}.
\end{theorem}

The proof of Theorem \ref{thm:enhanced_expertbound_cr} is deferred to Section~\ref{sec:prooftheorem3.2}. The key challenge is to lower bound the assigned frame-level budget $B^{(i)}$ in Eqn.~\eqref{eqn:allocatebudget} for frame $i$ and get an effective per-round budget $\hat{\rho}=\rho+\Delta\rho$. To do so, we construct an effective additive budget $\hat{\Omega}_i(\beta)$ for frame $i$ given any $\beta> 0$ in \eqref{eqn:effectivebudget} and prove that $\hat{\Omega}_i(\beta)$ is the infimum of the additive budget $\Omega_i(\beta)$ by \experttwo for any $\beta>0$. Then, by selecting the worst-case per-round effective reference budget $\rho+\hat{\Omega}_i/(2^{i-1}T^*)$ for each frame $i$ and optimizing it by choosing $\beta$, we obtain the  per-round budget $\hat{\rho}=\rho+\Delta \rho$. At last, by summing up the utility bounds of all the frames, the difference between the optimal utility and \experttwo is bounded as
$OPT(y)-\hat{\alpha} F_T^{\experttwo}(y)\leq \hat{C}_1+\hat{C}_2 \sqrt{T}$,
 where $\hat{\alpha}=\max_{m\in[M]}\frac{\bar{x}_m}{\rho_m+\Delta\rho_m}$,
and
$\hat{C}_1$ and $\hat{C}_2$ are two finite constants in Appendix~\ref{sec:prooftheorem3.2}. This is then translated to the asymptotic competitive ratio in Theorem~\ref{thm:enhanced_expertbound_cr}.

Different from the competitive ratio of \expert which relies on the fixed per-round budget $\rho$,
the competitive ratio of \experttwo utilizes the effective per-round budget $\rho+\Delta\rho$, which includes the fixed part
 $\rho$ and the additional part $\Delta\rho$ due to replenishment (subject to the maximum budget cap $B_{\max}$).
Importantly, $\Delta\rho$ is positive if the minimum replenishment over a unit frame $E_{\min}>0$, resulting in 
a higher asymptotic competitive ratio than \expert. When $E_{\min}=0$, there is no guarantee of minimum budget replenishment for each unit frame. Hence, the asymptotic competitive ratio of \ouralgtwo reduces to the one achieved by \expert in the worst case since we cannot rule out the case in which there is no budget replenishment at all. 
Thus, the improvement of the competitive ratio by \experttwo does not conflict with the optimality of the competitive ratio achieved by \expert for general cases (which includes the case of no budget replenishment). 

The insights of the asymptotic competitive ratio of \experttwo are further explained as follows. The improvement of the competitive ratio compared with \expert depends on $\Delta \rho_m$, which is lower bounded by the minimum of two terms. 
The first term $\frac{E_{\min,m}}{2T^*}$ indicates the effect of the minimum amount of budget replenishment within a unit frame. Naturally, a larger minimum budget replenishment can make the problem less resource-constrained and lead to a higher competitive ratio. The second term
in the minimum operation 
shows the effect of the maximum budget cap $B_{\max,m}$ on constraining the actual budget replenishment following \eqref{eqn:budget_dynamics}. 
The second term has a different expression for resource $m$ with $B_{\max,m}<(T+T^*)\rho_m$ because a small $B_{\max,m}$
can result in less space for replenishment.
No matter whether $B_{\max,m}\geq (T+T^*)\rho_m$ holds, a higher budget cap $B_{\max}$ 
allows more budgets to be replenished, thus leading to a higher competitive ratio. 
If the budget cap $B_{\max}$ is large enough, it does not constrain the budget replenishment any more and the competitive ratio improvement only depends on the minimum budget replenishment $E_{\min}$.  
 In addition, the best choices of threshold hyper-parameter $\beta_m$  increases with $\rho_{\max,m}=B_{\max,m}/T$. This is consistent with the intuition that with a larger budget cap $B_{\max,m}$, the threshold of the additive budget in Eqn.~\eqref{eqn:allocatebudget} can be set larger to assign the frame-level budget more aggressively.  
These observations all confirm the intuition that a larger budget cap
can utilize the budget replenishment more effectively, increasing the asymptotic competitive ratio.


\section{\ouralg: Learning-Augmented Online Allocation}\label{sec:learning}

While \expert and \experttwo have provable worst-case performance
guarantees (in terms of asymptotic competitive ratio),
they may not perform well on average due to their conservativeness
in resource pricing $\mu_t$ in order to address the worst-case uncertainties
in budget replenishments. 
In this section, we go beyond the worst-case and
propose a novel learning-augmented approach, \ouralg,
that integrates an ML-based online optimizer
with \expert (or \experttwo) to improve the average performance (Theorem~\ref{thm:average}) while still being able to guarantee the worst-case performance (Theorem~\ref{thm:robustness}).

\subsection{Average Utility Maximization with Worst-Case Utility Constraint}\label{sec:learning_objective}

We first present our optimization objective of designing a learning-augmented online algorithm $\pi$ as follows --- maximizing the average utility subject to a worst-case utility constraint. Since the competitive algorithms (i.e., \expert or \experttwo) have been proved to ensure the asymptotic competitive ratios, we guarantee the worst-case utility of the learning-augmented online algorithm $\pi$ by comparing it with the utility of a competitive algorithm. Thus, the objective of our learning-augmented online algorithm is 
\begin{subequations}\label{eqn:objective_learning_formulation}
 \begin{gather}\label{eqn:objective_learning}
     \max_{\pi} \mathbb{E}_{y}\left[F_T^{\pi}(y)\right]\\ 
       \label{eqn:constraint_1_learning}
       s.t., \;\;\;\; F_T^{\pi}(y)\geq \lambda F_T^{\pi^{\dagger}}(y)-R,\;\;\;\forall y\in\mathcal{Y}, 
       \end{gather}
\end{subequations}
where  $F_T^{\pi}(y)=\sum_{t=1}^Tf(x_t,c_t)$
is the total utility of an online algorithm $\pi$,
$\lambda\in[0,1]$ represents multiplicative competitiveness
of the online algorithm $\pi$ with respect
to the algorithm $\pi^{\dagger}$ (i.e., \expert or \experttwo in our case) and $R\geq0$ indicates the additive slackness in
the utility constraint. 
Note that considering a sequence-wise distribution
of 
$y\sim \mathbb{P}_y$ differs from the standard stochastic
setting where each online input $y_t$ for $t\in[t]$ is assumed to follow an 
independent and identically distributed (i.i.d.) distribution
(e.g., i.i.d. utility function $f_t$ in \cite{OnlineAllocation_DualMirroDescent_Google_OperationalResearch_2022_doi:10.1287/opre.2021.2242}, or i.i.d. potential replenishment $\hat{E}_t$ in \cite{OnlineOpt_EnergyHarvesting_Bregman_Neely_USC_SIgmetrics_2020_10.1145/3428337}), because $y_t$ for $t\in[t]$ within an sequence
can still be arbitrary in our problem \eqref{eqn:objective_learning_formulation}.

The parameters $\lambda\in[0,1]$ and $R\geq0$ can be viewed
as worst-case robustness requirement with respect
to \expert or \experttwo (denoted
as $\pi^{\dagger}$ for the convenience of presentation).
Concretely, when $\lambda\in[0,1]$ increases and/or $R\geq0$ decreases,
the online algorithm $\pi$ is closer to $\pi^{\dagger}$ in terms
of the worst-case utility, and vice versa. 
 Moreover, 
as $\pi^{\dagger}$ itself has performance guarantees
 and
is asymptotically competitive against the optimal oracle $OPT$ (Theorem~\ref{thm:expertbound} and Theorem~\ref{thm:enhanced_expertbound_cr}),
the constraint in \eqref{eqn:constraint_1_learning} also immediately
translates into provable asymptotic competitiveness of the online algorithm $\pi$
with respect to $OPT$. 
That is, given the asymptotic competitive ratio $CR^{\pi^{\dagger}}$ achieved by $\pi^{\dagger}$, the constraint \eqref{eqn:constraint_1_learning} leads to $\lim_{T\rightarrow \infty}\sup_y \frac{1}{T}\left(OPT(y)-\frac{1}{\lambda CR^{\pi^{\dagger}}}F_T^{\pi}(y)\right)\leq \lim_{T\rightarrow \infty}\sup_y \frac{1}{T}\left(OPT(y)-\frac{1}{CR^{\pi^{\dagger}}}F_T^{\pi^{\dagger}}(y)+\frac{R}{\lambda CR^{\pi^{\dagger}}}\right) \leq 0$, guaranteeing an asymptotic competitive ratio of $\lambda \cdot CR^{\pi^{\dagger}}$ for $\pi$.
In fact, considering a baseline algorithm for
worst-case robustness is also a common practice in  existing
 learning-augmented algorithms \cite{SOCO_ML_ChasingConvexBodiesFunction_Adam_COLT_2022,SOCO_OnlineOpt_UnreliablePrediction_Adam_Sigmetrics_2023_10.1145/3579442,Control_RobustConsistency_LQC_TongxinLi_Sigmetrics_2022_10.1145/3508038}. 
 Thus, in the following,
it suffices to consider \eqref{eqn:objective_learning_formulation}
to achieve the best of both worlds: maximizing the average utility
while bounding the worst-case utility (directly with respect to \expert or \experttwo and also indirectly with
respect to $OPT$).

Unlike \expert or \experttwo that is particularly designed to address
the worst-case robustness, an ML model can readily
exploit statistical information of $y\in\mathcal{Y}$ based on history instances.
Thus, one may want
to use a pure ML-based online optimizer to maximize
the average utility for solving 
\eqref{eqn:objective_learning_formulation}. Nonetheless,
ML-based optimizers typically do not have worst-case performance guarantees
as hand-crafted algorithms (\expert or \experttwo in our case) due to, e.g., training-testing distributional shifts. 
In fact, even by assuming perfect ML-based optimizers,
maximizing the average utility alone does not necessarily guarantee
the worst-case robustness in \eqref{eqn:constraint_1_learning}.
The reason is that maximizing the average utility
needs to prioritize many typical problem instances,
while the worst-case robustness needs to address those rare but possible
corner cases. In general, the trade-off
between average utility and worst-case robustness is unavoidable
and well-known for online optimization problems, thus
spurring the emerging field of learning-augmented online algorithms that leverage both ML predictions and hand-crafted algorithms (see, e.g., \cite{SOCO_ML_ChasingConvexBodiesFunction_Adam_COLT_2022,SOCO_OnlineOpt_UnreliablePrediction_Adam_Sigmetrics_2023_10.1145/3579442,OnlineOpt_Learning_Augmented_RobustnessConsistency_NIPS_2020,SOCO_MetricUntrustedPrediction_Google_ICML_2020_pmlr-v119-antoniadis20a} for studies in other online problems).

\subsection{Algorithm Design}\label{sec:learning_algorithm}

We now present the design of \ouralg, a novel learning-augmented algorithm
for online allocation with replenishable budgets 
under an additional mild assumption of
Lipschitz utility functions.
\begin{assumption}[Lipschitz utility]\label{asp:lipschitz}
 For any $t\in[T]$, the utility function $f_t(x)$ is $L$-Lipschitz continuous with respect to $x$, i.e. $\forall x,x'\in\mathcal{X}$, 
 we have $\left|f_t(x)-f_t(x')\right|\leq L\|x-x'\|$, where $L>0$
 and $\|\cdot\|$ is a norm operator. 
\end{assumption}

The Lipschitz assumption implies a bounded utility change
given a bounded action change, which is reasonable for real applications and commonly assumed in online problems.
Remember that to guarantee a competitive ratio, \expert in
Algorithm~\ref{alg:expert} and \experttwo in Algorithm~\ref{alg:expert_2} conservatively set their resource
 prices $\mu_t$ in two different conservative manners. Thus,
 the key goal of \ouralg is to overcome the conservativeness of competitive algorithms like \expert and \experttwo
by exploiting the distribution of $y\in\mathcal{Y}$ to improve the average utility
while bounding the worst-case utility loss with respect to \expert/\experttwo.\footnote{For notational simplicity, we use \ouralg to represent our learning-augmented algorithm, noting
that the competitive algorithm used
by \ouralg can also be \experttwo.}
Towards this end, \ouralg utilizes an ML policy/predictor (denoted
as $\tilde{\pi}$) as well as a competitive algorithm (denoted as $\pi^{\dagger}$) 
 that output their decisions as advice,
and then
judiciously chooses the actual online decisions.

Naturally, always following the decisions of competitive algorithm
satisfies
the worst-case utility constraint in \eqref{eqn:constraint_1_learning}, but 
fails to utilize ML for average utility improvement. On the other hand, blindly
following the ML policy 
can potentially improve the average performance but the worst-case utility constraint is not guaranteed. 

Thus,  a key challenge of learning-augmented online algorithms is how to utilize
the decisions of the ML policy and a worst-case robust algorithm (i.e., \expert and \experttwo in our case)
as online advice \cite{SOCO_ML_ChasingConvexBodiesFunction_Adam_COLT_2022,SOCO_OnlineOpt_UnreliablePrediction_Adam_Sigmetrics_2023_10.1145/3579442}.
To address this challenge,
given $\tilde{x}_t$ and $x_t^{\dagger}$
that represent the allocation decisions by the ML policy and the competitive algorithm, respectively,
 \ouralg chooses the actual decision $x_t$ using a novel 
 reservation utility which we introduce as follows. In the following,
to be consistent with the literature \cite{SOCO_ML_ChasingConvexBodiesFunction_Adam_COLT_2022},
we also refer to the ML policy's decision $\tilde{x}_t$ as ML predictions.

\textbf{Constrained decision set.} To ensure that an online algorithm $\pi$ satisfies the worst-case utility constraint \eqref{eqn:constraint_1_learning} for any sequence $y\in\mathcal{Y}$,
it might seem sufficient to guarantee $\sum_{\tau=1}^tf_t(x_t)\geq \lambda \sum_{\tau=1}^tf_t(x_t^{\dagger})-R$ for each round $t\in[T]$. Nonetheless,
even though the constraint $\sum_{\tau=1}^tf_t(x_t)\geq \lambda \sum_{\tau=1}^tf_t(x_t^{\dagger})-R$ is satisfied for round $t$, it may not
be guaranteed at round $t+1$, thus potentially violating the worst-case utility
constraint at the end of the sequence. 
Let us now consider an illustrative example to explain this point. Suppose that the algorithm $\pi$ satisfies  $\sum_{\tau=1}^tf_t(x_t)\geq \lambda \sum_{\tau=1}^tf_t(x_t^{\dagger})-R$ but
allocates more resources than $\pi^{\dagger}$ up to round $t$. Then, in future rounds,
it is possible that there is very little budget replenishment and the algorithm
$\pi^{\dagger}$ can still allocate resources to gain a higher utility, whereas the algorithm $\pi$ does not have enough resources to match the utility of 
$\pi^{\dagger}$. In other words, if $\pi$ uses more resources
than $\pi^{\dagger}$ up to round $t$,
 the satisfaction of utility constraint by $\pi$ in terms
of $\sum_{\tau=1}^tf_t(x_t)\geq \lambda \sum_{\tau=1}^tf_t(x_t^{\dagger})-R$
is just \emph{temporary} and can still be violated in the future.

To address such uncertainties in the future, 
we introduce a novel reservation utility
$\Delta(x_t)= \lambda L \sum_{m=1}^M\left[(B_{t}^{\dagger}+E_t^{\dagger}  - x_t^{\dagger})_m-(B_{t}+E_t - x_t )_m\right]^+$ into the utility constraint for each round
$t$, where $(B_{t}+E_t - x_t )_m$ means the remaining budget
for the type-$m$ resource at the end of round $t$. The interpretation of $\Delta(x_t)$
is to bound
the maximum potential utility advantage (scaled by $\lambda\in[0,1]$) obtained by $\pi^{\dagger}$ in future rounds, if $\pi^{\dagger}$ has more remaining budgets
compared to $\pi$ at the end of round $t$; on the other hand, if the algorithm $\pi$
has even more resources available than  $\pi^{\dagger}$, there is no need to add the reservation since  $\pi$ can always roll back to the decision of $\pi^{\dagger}$ in the future without worrying about budget shortages. Here, we simply use
$\Delta(x_t)$ for the convenience of presentation while suppressing its dependency on other terms such as $x_t^{\dagger}$.
Thus, by adding $\Delta(x_t)$, we now have a new constraint on the decision $x_t$
as follows:
\begin{equation}\label{eqn:constraint_new_reservation}
\sum_{i=1}^t f_t(x_i)\geq \lambda \sum_{i=1}^t f_t(x_i^{\dagger}) +\Delta(x_t)-R,
\end{equation}
which, if satisfied at round $t$, guarantees the existence of at least
one feasible decision that can still satisfy the constraint. In other words,
if the decisions $x_t$ are chosen out of the constrained set \eqref{eqn:constraint_new_reservation} for round $t\in[T]$,
worst-case utility constraint \eqref{eqn:constraint_1_learning} can be satisfied 
 at the end of any sequence $y\in\mathcal{Y}$.
To our knowledge, the design of $\Delta(x_t)$ for
constructing a constrained decision set \eqref{eqn:constraint_new_reservation}
is novel for online allocation with replenishable budgets
and also differs from many prior learning-augmented algorithms (e.g.,
\cite{SOCO_MetricUntrustedPrediction_Google_ICML_2020_pmlr-v119-antoniadis20a}
uses a pre-determined threshold for dynamically switching between ML prediction
$\tilde{x}_t$
and the worst-case robust action $x_t^{\dagger}$).

\begin{algorithm}[!t]
	\caption{Learning-Augmented Online Allocation with Replenishable Budgets (\ouralg)}
	\begin{algorithmic}[1]\label{alg:online_inference}
		\REQUIRE ML policy $\tilde{\pi}$ and the competitive algorithm $\pi^{\dagger}$ (\expert or \experttwo)
		\FOR   {$t=1$ to $T$}
		\STATE Receive reward function $f_t$, and potential budget replenishment $\hat{E}_t$.
        \STATE Get replenished budgets $E_t=\min\{\hat{E}_t,B_{\max}-B_t\}$
        for \ouralg, and $E_t^{\dagger}=\min\{\hat{E}_t,B_{\max}-B_t^{\dagger}\}$ for $\pi^{\dagger}$
        \STATE Get ML prediction $\tilde{x}_t$\\ 
       \STATE Get the action $x_t^{\dagger}$ of $\pi^{\dagger}$
       based on its own history (by Algorithm~\ref{alg:expert} or Algorithm~\ref{alg:expert_2})\\
       \STATE Choose $x_t$ by solving
       \begin{subequations}\label{eqn:projection_algorithm}
 \begin{gather}\label{eqn:projection_objective}
     x_t=\arg\min_{x\in\mathcal{X}}\|x-\tilde{x}_t\|\\ 
       \label{eqn:projection_constraint}
       s.t., \sum_{i=1}^t f_t(x_i)\geq \lambda \sum_{i=1}^t f_t(x_i^{\dagger}) +\Delta(x_t)-R, \text{ and }
         x_t \leq B_{t}+E_t, 
       \end{gather}
\end{subequations}
       where $\Delta(x_t)= \lambda L \sum_{m=1}^M\left[(B_{t}^{\dagger}+E_t^{\dagger}  - x_t^{\dagger})_m-(B_{t}+E_t - x_t )_m\right]^+$
		\STATE Update budgets $B_{t+1}=B_{t}
   -x_t+E_t$ for \ouralg, and $B^{\dagger}_{t+1}=B^{\dagger}_{t}
   -x_t^{\dagger}+E_t^{\dagger}$ for $\pi^{\dagger}$
	\ENDFOR
	\end{algorithmic}
\end{algorithm}

\textbf{Algorithm.} Next, we describe the online optimization process of \ouralg in Algorithm~\ref{alg:online_inference}. In \ouralg, the competitive algorithm (i.e., $\pi^{\dagger}$) runs independently
for the purpose of bounding the worst-case utility constraint \eqref{eqn:constraint_1_learning}, and the ML predictor $\tilde{\pi}$ takes
the actual online information $y_{1:t}$ (including the actual remaining budget $B_t$
and replenishment $E_t$) as its input and generates its prediction $\tilde{x}_t$
as advice to \ouralg. Then, $\tilde{x}_t$ is projected into a constrained decision
set \eqref{eqn:constraint_new_reservation} to find
the actual decision $x_t$ that guarantees the worst-case utility constraint.
The purpose of the projection in \ouralg is
to ensure that $x_t$ is both close to the ML prediction $\tilde{x}_t$ to exploit its potential for improving the average utility,
while
still staying inside the constrained decision set \eqref{eqn:constraint_new_reservation}
for worst-case utility constraint \eqref{eqn:constraint_1_learning}.

\textbf{ML training.} Up to this point, we have assumed that the ML predictor/policy $\tilde{\pi}$ has been provided to \ouralg for online optimization. Next, we discuss how
to train $\tilde{\pi}$ used in Algorithm \ref{alg:online_inference}. In the context of online optimization, the ML-based predictor/policy
is typically trained offline and then applied online for inference \cite{L2O_OnlineBipartiteMatching_Toronto_ArXiv_2021_DBLP:journals/corr/abs-2109-10380,L2O_NewDog_OldTrick_Google_ICLR_2019,L2O_AdversarialOnlineResource_ChuanWu_HKU_TOMPECS_2021_10.1145/3494526,Shaolei_L2O_ExpertCalibrated_SOCO_SIGMETRICS_Journal_2022}.
Here, we adopt this standard practice for \ouralg. Specifically,
we collect a training dataset $\mathcal{S}$ of episodic information $y\in\mathcal{Y}$ based on history data and/or data augmentation techniques, and build an ML model (e.g., a recurrent neural network for online sequential decision making, with each base network
parameterized by the same weights \cite{L2O_NewDog_OldTrick_Google_ICLR_2019,L2O_OnlineBipartiteMatching_Toronto_ArXiv_2021_DBLP:journals/corr/abs-2109-10380}).

We train the ML model $\tilde{\pi}$ by optimizing the expected utility obtained by Algorithm~\ref{alg:online_inference}. Denote \ouralg($\tilde{\pi}$) as the algorithm \ouralg with the ML model $\tilde{\pi}$. The training objective is expressed as 
\begin{equation}\label{eqn:training_obj}
\max_{\tilde{\pi}}\frac{1}{|\mathcal{S}|}\sum_{y\in\mathcal{S}} F_T^{\ouralg(\tilde{\pi})}(y),
\end{equation}
where $F_T^{\ouralg(\tilde{\pi})}(y)$ is the total utility of $\ouralg(\tilde{\pi})$ for the sequence $y$.

 To train the ML model, we apply the state-of-the-art backpropagation, while noting that 
differentiation of the projection operator (which itself is a constrained optimization problem) with respect to the ML 
prediction $\tilde{x}_t$ is needed and can be performed based on 
implicit differentiation techniques \cite{DNN_ImplicitLayers_Zico_Website,L2O_DifferentiableMPC_NIPS_2018_NEURIPS2018_ba6d843e,L2O_DifferentiableOptimization_Brandon_ICML_2017_amos2017optnet}.

\subsection{Performance Analysis}\label{sec:learning_analysis}

We now present the performance analysis of \ouralg in terms of its worst-case 
utility as well as its average performance. As formally stated below, our results highlight
that \ouralg guarantees the worst-case utility constraint for any sequence
$y\in\mathcal{Y}$ and meanwhile is able to exploit the benefits of ML predictions
to improve the average utility.

\subsubsection{Worst-Case Utility} We first present
the worst-case utility of \ouralg.

\begin{theorem}\label{thm:robustness}
For any $\lambda\in[0,1]$
and $R\geq0$, given any ML predictor $\tilde{\pi}$ and by the design $\Delta(x_t)= \lambda L \sum_{m=1}^M\left[(B_{t}^{\dagger}+E_t^{\dagger}  - x_t^{\dagger})_m-(B_{t}+E_t - x_t )_m\right]^+$,
\ouralg in Algorithm~\ref{alg:online_inference}
 always guarantees the worst-case utility constraint  \eqref{eqn:constraint_1_learning} for any sequence $y\in\mathcal{Y}$.
\end{theorem}

The proof of Theorem~\ref{thm:robustness} is available in Appendix~\ref{appendix:proof_robustness} and shows
that, based on the design of $\Delta(x_t)$, if \eqref{eqn:constraint_new_reservation} is satisfied for round $t$,
then there must always exist a feasible solution satisfying \eqref{eqn:constraint_new_reservation} for round $t+1$.

Theorem~\ref{thm:robustness} guarantees that the worst-case utility
constraint \eqref{eqn:constraint_1_learning} is always satisfied for any sequence $y\in\mathcal{Y}$ regardless
of how bad the ML predictions are.
Thus, even
when the training-testing distributions differ and/or the ML predictions
are adversarially modified, \ouralg can still offer worst-case utility
guarantees with respect to the competitive algorithm \expert or \experttwo.

\subsubsection{Average Utility}\label{sec:average_utility}

Besides
the robustness guarantee, the performance of
a learning-augmented algorithm is often analyzed by considering
the \emph{worst}-case competitive ratio (a.k.a., consistency) under the assumption that ML predictions are perfect and offline optimal for any sequence $y\in\mathcal{Y}$ \cite{SOCO_MetricUntrustedPrediction_Google_ICML_2020_pmlr-v119-antoniadis20a,OnlineOpt_Learning_Augmented_RobustnessConsistency_NIPS_2020}.
The consistency metric measures how closely a learning-augmented algorithm
can follow the perfect ML predictions in the worst case. 
However, an ML model in practice is typically not perfect and instead
is trained to maximize the average performance in practice.
Thus, to measure the capability of \ouralg for following
ML predictions, we directly bound the average utility of \ouralg and compare it with the average utility of the optimal unconstrained ML model $\tilde{\pi}^*=\arg\max_{\pi} \mathbb{E}_{y}\left[F_T^{\pi}(y)\right]$ that provides the best average performance. As such, 
given an optimally
trained ML model,
we measure the \emph{average}-case consistency of \ouralg
with respect to the optimal unconstrained ML model $\tilde{\pi}^*$ 
in terms of the average utility difference between \ouralg and $\tilde{\pi}^*$.
Our consideration of an optimal ML model 
essentially parallels the assumption of ``perfect ML prediction'' for worst-case consistency analysis
of learning-augmented algorithms \cite{OnlineOpt_Learning_Augmented_RobustnessConsistency_NIPS_2020,SOCO_OnlineOpt_UnreliablePrediction_Adam_Sigmetrics_2023_10.1145/3579442}, 
while noting that our optimality is in the average sense subject
to our designed constrained decision set \eqref{eqn:constraint_new_reservation}.

More concretely, 
we consider an optimal ML predictor that optimizes the average utility of \ouralg, i.e. \begin{equation}\label{eqn:opt_constrained_ml}
\tilde{\pi}^{\circ}=\arg\max_{\tilde{\pi}} \mathbb{E}_{y}\left[F_T^{\ouralg(\tilde{\pi})}(y)\right],
\end{equation}
where $\ouralg(\tilde{\pi})$ outputs the actions $x_t$ satisfying
\eqref{eqn:constraint_new_reservation} given the ML prediction $\tilde{x}_t$, 
and show the average utility bound of \ouralg in the next theorem.

\begin{theorem}\label{thm:average}
For any $\lambda\in[0,1]$
and $R\geq0$, 
the average utility of \ouralg with the optimal ML model $\tilde{\pi}^{\circ}$ is bounded by
\begin{equation}\label{eqn:consistency bound}
\mathbb{E}_y\left[F_T^{\ouralg(\tilde{\pi}^{\circ})}(y)\right]\geq \max\left\{\mathbb{E}_{y}\left[F_T^{\tilde{\pi}^*}(y)\right]-L(1-\gamma_{\lambda,R})\mathbb{E}_{y}\left[\sum_{t=1}^T \|x_t^{\dagger}-\tilde{x}_t^{*}\|\right], \mathbb{E}_{y}\left[F_T^{\pi^{\dagger}}(y)\right]\right\}
\end{equation}
where  $\gamma_{\lambda, R}=\min\left\{1, \frac{R}{2\lambda L\theta}\right\}$, 
$\theta=\max_{y}\sum_{t=1}^T\|\tilde{x}_t^*-x_t^{\dagger}\|_1$
indicates the maximum cumulative decision difference between the action $x_t^{\dagger}$ of \expert or \ouralgtwo and the action $\tilde{x}_t^*$ of the optimal unconstrained ML predictor
 $\tilde{\pi}^*$, 
and $L$ is the Lipschitz constant of
the utility function (Assumption~\ref{asp:lipschitz}). 
\end{theorem}

The proof of Theorem~\ref{thm:average} is available in Appendix~\ref{appendix:proof_consistency}.
The key idea is to translate the constraint \eqref{eqn:constraint_new_reservation}
into a new distance constraint between $x_t$ and
$x_t^{\dagger}$. Thus, if $x_t$ is sufficiently
close to $x_t^{\dagger}$ for each round $t\in[T]$, we guarantee the worst-case utility constraint \eqref{eqn:constraint_1_learning}.
Meanwhile, by considering the optimal unconstrained ML predictor
$\tilde{\pi}^*=\arg\max_{\pi} \mathbb{E}_{y}\left[F_T^{\pi}(y)\right]$,
we find the closest distance between 
$x_t$ and the ML prediction $\tilde{x}_t^*$ subject to the distance constraint between $x_t$ and
$x_t^{\dagger}$, and use this as a feasible online algorithm. The bound of such a closest distance requires an analysis of the remaining budget perturbation depending on the non-linear budget dynamics in \eqref{eqn:budget_dynamics} due to the maximum budget cap.
Next, by optimality of $\tilde{\pi}^{\circ}$ used by \ouralg
to explicitly maximize the average utility satisfying
our constraint \eqref{eqn:constraint_new_reservation},
we obtain the bound in Theorem~\ref{thm:average}.

Theorem~\ref{thm:average} shows that the average utility of $\ouralg(\tilde{\pi}^{\circ})$ with the optimal ML model $\tilde{\pi}^{\circ}$ is no worse than that of the competitive algorithm $\pi^{\dagger}$ (\expert or \experttwo) which is the second term in the maximum operator.
The reason is that the competitive algorithm $\pi^{\dagger}$ is one of the decision policies with actions in the constrained decision sets \eqref{eqn:constraint_new_reservation}, whereas $\ouralg(\tilde{\pi}^{\circ})$ is the optimal policy satisfying the decision constraints \eqref{eqn:constraint_new_reservation}.
This indicates that, while providing the worst-case performance guarantee, \ouralg can still improve the average utility compared with the competitive algorithm (\expert or \experttwo). The improvement relies on the first term in the maximum operator, which
bounds the average utility difference between \ouralg and the 
optimal-unconstrained ML model $\tilde{\pi}^*$.

The first term in the maximum operator in Theorem~\ref{thm:average}
provides the key insight into the tradeoff between the worst-case performance and average performance.
Specifically, with a smaller $\lambda\in[0,1]$ 
and/or greater $R\geq0$, the worst-case utility constraint is less stringent 
and hence provides
more flexibility for \ouralg to exploit the benefits of ML predictions for higher average utility,
and vice versa. In particular, when $R$ is large enough or $\lambda$ is small enough, the worst-case utility constraint in \eqref{eqn:constraint_1_learning} is so relaxed that it does not affect average utility maximization. In such cases, \ouralg approaches the average utility of the optimal unconstrained ML predictor. 
 When the 
decisions of the optimal-unconstrained
ML predictor and the competitive algorithm become more distinct (i.e., increasing $\theta$ or $\mathbb{E}_{y}\left[\sum_{t=1}^T \|x_t^{\dagger}-\tilde{x}_t^{*}\|\right]$ in Theorem~\ref{thm:average}), it is natrually
more difficult to follow the ML predictions while still
staying close to the competitive algorithm for worst-case utility,
unless we lessen the worst-case utility constraint by decreasing
$\lambda\in[0,1]$ and/or increasing $R\geq0$.

Theorem~\ref{thm:average} gives the bound of average utility by assuming an optimal ML model $\tilde{\pi}^{\circ}$ which parallels the assumption of    ``perfect ML prediction''
for the worst-case consistency analysis in existing learning-augmented algorithms \cite{OnlineOpt_Learning_Augmented_RobustnessConsistency_NIPS_2020,SOCO_OnlineOpt_UnreliablePrediction_Adam_Sigmetrics_2023_10.1145/3579442}.  
However, if the ML model $\tilde{\pi}$ in \ouralg is not optimally trained, we can define the ML prediction imperfectness as $\epsilon = \mathbb{E}_{y}\left[F_T^{\tilde{\pi}^{\circ}}(y)-F_T^{\tilde{\pi}}(y)\right]$, where $\tilde{\pi}^{\circ}=\arg\max_{\tilde{\pi}}\mathbb{E}_y\left[F_T^{\ouralg(\tilde{\pi})}(y)\right]$ is the optimal ML model for \ouralg. The imperfectness
can come from a variety of sources, including
finite model capacity and potential training-testing distributional shifts.
Then, the average utility bound with respect to $\tilde{\pi}$ can be obtained by subtracting the ML imperfectness $\epsilon$ from the average utility bound in Theorem~\ref{thm:average}.
 Nevertheless, even when $\epsilon\to\infty$,
the average utility of \ouralg is always bounded by 
$\lambda \mathbb{E}_y\left[F_T^{\pi^{\dagger}}(y)\right]-R$, 
where $\mathbb{E}_y\left[F_T^{\pi^{\dagger}}(y)\right]$ is the average utility
of the competitive algorithm (\expert or \experttwo) used by \ouralg.
This is a natural
byproduct of Theorem~\ref{thm:robustness}, which guarantees the worst-case utility constraint of \ouralg with respect to the competitive algorithm.

In general, achieving
the optimal tradeoff between average utility and
the worst-case utility is extremely challenging for learning-augmented algorithms
(see, e.g., \cite{SOCO_OnlineOpt_UnreliablePrediction_Adam_Sigmetrics_2023_10.1145/3579442,SOCO_ML_ChasingConvexBodiesFunction_Adam_COLT_2022} for discussions on smoothed online convex optimization). Nonetheless, although it remains an open problem to achieve
the best tradeoff,
our result in Theorem~\ref{thm:average} provides the first characterization
of such a tradeoff in the context of learning-augmented algorithms
for online allocation with budget replenishment. 
In fact, 
even a competitive online algorithm with budget replenishment is lacking prior
to our design of \expert and \experttwo.

\section{Simulation Study}\label{sec:simulation}

In this section, we run a simulation study on sustainable AI inference powered by renewables. 
First, we present the experimental setup, followed by the comparative analysis of the results from our algorithms with 
existing baselines. We show that \ouralg 
has improved performance in terms of average utility while still being able to offer
good worst-case utility.

\subsection{Setup}
This section presents our problem setting, dataset, baseline algorithms, and  ML model architecture.

\textbf{Problem setting.} 
Edge data centers are becoming a major platform for AI inference thanks
to their proximity to end users.
To achieve sustainable AI inference on the edge, it is important to 
exploit renewable generation to replenish on-site energy storage. This can
significantly lower the carbon emissions caused by the surging  demand for AI inference  \cite{ML_Carbon_Bloom_176B_Sasha_Luccioni_arXiv_2022_luccioni2022estimating}. 
For a given AI inference service, multiple models are often available. For instance, there are eight different GPT-3 models  \cite{ML_GPT3_Energy_Others_NIPS_2020_NEURIPS2020_1457c0d6}, each with distinct model sizes, providing a flexible balance between accuracy and energy consumption.
However, the renewable sources are known for their time-varying and unstable nature. Thus, 
we can
use intermittent renewables 
to replenish the energy budgets,
and schedule an appropriate AI model
 for inference in an online manner to maximize the utility given
available energy budget constraints \cite{carbon_aware_computing_radovanovic2022carbon,green_AI_schwartz2020green}.

Specifically, we focus on an edge data center with
an on-site energy storage unit (e.g., batteries) for AI inference.
The initial energy budget is $B_1=12 kWh$. At each round $t$, the time-varying renewable energy $E_t$ is replenished to the energy storage subject to the maximum capacity constraint $B_{\max}=30 kWh$. Each problem instance has 120 rounds. If served by the full AI model,
the energy consumption for inference is $c_t$, which also measures
the total demand. Nonetheless,
the resource manager can decide an AI model at each round $t$, which consumes energy $x_t$. If a smaller AI model is chosen, then $x_t$ is also smaller,
but the inference accuracy is potentially lower. Here, we use
a utility function to denote the reward by consuming $x_t$
energy for serving the demand. Specifically,
we model the utility of serving each unit of AI inference demand as $\log (1+ \min\{1,\frac{x_t}{c_t}\})$, where the min operator means that over-using energy $x_t$ beyond the maximum demand does not offer additional utility.
Next, by using the total demand $c_t$ to scale the demand,
we have a utility function of $f_t(x_t)=c_t\log (1+ \min\{1,\frac{x_t}{c_t}\})$ at time $t$. Note that choosing $x_t=0$ means that the inference demand is
not processed by the edge (and routed to cloud data centers beyond our scope).
The remaining budget in the energy storage is then updated according to \eqref{eqn:budget_dynamics}. The goal of the resource manager 
is to maximize $\sum_{t=1}^Tf_t(x_t)$ subject to the energy budget constraint.

\textbf{Dataset.}  In our experiment, the inference demand $c_t$ comes from the GPU power usage of the BLOOM model (a large lanugage model) API running on 16 Nvidia A100 GPUs \cite{ML_Carbon_Bloom_176B_Sasha_Luccioni_arXiv_2022_luccioni2022estimating}.
The budget replenishment  $E_t$ 
(harvested renewable energy) is constructed based on the renewable dataset from California Independent System Operator \cite{caiso}, which
contains hourly solar renewables. The values are scaled down to our setting. We extend the BLOOM trace data by data augmentation to construct a training dataset consisting of 1600 problem instances, each with 120 hours. 
Then, the entire dataset is divided  into  training and testing 
sets with a 3:1 ratio.

\textbf{Baseline algorithms.}
To compare our results, we consider the following baseline algorithms.

-- \emph{OPT}: OPT is the optimal oracle algorithm that solves \eqref{eqn:objective_offline}
based on complete offline information. Thus, OPT has the highest utility for any problem instance.

-- \emph{Equal:} Equal uniformly allocates the initial budget, and greedily uses the replenished budget whenever applicable, i.e., $x_t=\min\{\bar{x},\rho+E_t\}$.

-- \emph{Greedy:} Greedy allocates as much budget as possible at each round, i.e., $x_t=\min\{\bar{x},B_t+E_t\}$.

-- \emph{DMD:} DMD (Dual Mirror Descent) updates the dual variable by mirror descent \cite{OnlineAllocation_DualMirroDescent_Google_OperationalResearch_2022_doi:10.1287/opre.2021.2242}. With replenishable budgets, {DMD} updates the dual variable based on  subgradient $g_t=\rho+E_t-\hat{x}_t$.

-- \emph{ML:} {ML} uses a standalone ML predictor to yield online allocation decisions subject to per-round budget constraints.
Such ML-based online optimizer  empirically have superior \emph{average} performance in a variety of online problems (when training-testing distributions are consistent) \cite{L2O_AdversarialOnlineResource_ChuanWu_HKU_TOMPECS_2021_10.1145/3494526,L2O_OnlineBipartiteMatching_Toronto_ArXiv_2021_DBLP:journals/corr/abs-2109-10380,L2O_NewDog_OldTrick_Google_ICLR_2019}, but cannot guarantee worst-case utility bounds.

The hyperparameters for these algorithms, if applicable, are tuned
based on our validation dataset to achieve the maximum utility.
While it is not possible to compare our algorithms with all
the existing baselines in the literature, our choice of baseline algorithms is representative in
the sense that they cover the strongest OPT, naive Greedy, state-of-the-art
competitive online algorithm DMD, as well as state-of-the-art ML-based online optimizers. Thus, we do not consider
other competitive algorithms than state-of-the-art DMD, or
other algorithms that focus on average performance 
but do not have as empirically good performance as ML.
Importantly, our design of \expert or \experttwo is provably-competitive and
\ouralg can provably satisfy the worst-case utility constraint \eqref{eqn:constraint_1_learning} with respect
to any available online algorithm by using it to replace \expert or \experttwo
as $\pi^{\dagger}$ in Algorithm~\ref{alg:online_inference}.

\textbf{ML model architecture.}
We implement the ML model based on a neural network with $2$ hidden layers each having a width of $10$ with ReLu activation. To train the model, we use the Adam optimizer for 100 epochs with a batch size of 20 and  a learning rate of $0.001$. The same ML architecture is also used in \ouralg.

\subsection{Results} 

In this section, 
we present a comparative analysis of different baselines with our proposed algorithms in terms of the average utility and empirical competitive ratio. The average utility is empirically calculated as the average utility of the testing samples and is normalized by the optimal average utility. The competitive ratio is empirically calculated as the minimum ratio of an online algorithm's utility to the optimal utility among the testing samples.
{Because of the provably better asymptotic competitive
ratio of \experttwo, we use \experttwo in \ouralg and set $R=0$ in \eqref{eqn:constraint_1_learning}} by default. All the utility values are normalized with respect to that of OPT.

\textbf{Comparison with baselines.} 
We first compare \expert, \experttwo and \ouralg with
the baseline algorithms in Table~\ref{table:result} under an
\emph{in-distribution} case where
the training-testing instances are drawn from the same distribution. 
Our results show
that ML has the highest average utility among the considered online algorithms,
while \ouralg, \expert, and \experttwo outperform other baselines (DMD, Greedy and Equal)
in terms of the average utility. Importantly, by setting $\lambda=0.3$
and $\lambda=0.6$, the average utilities of \ouralg are both improved 
compared to \expert and \experttwo, and closer to that of ML.

For the in-distribution testing case, the empirical
competitive ratio of ML is also the best, although ML does not have a guaranteed
competitive ratio. 
Besides, \expert and \experttwo both have higher competitive ratios than other baselines (DMD, Greedy, Equal), demonstrating their advantages in competitive ratio guarantees. Note that the empirical competitive ratios of \expert are higher than
 that of DMD which sets its resource price more aggressively, showing the benefit
of conservative pricing in \expert. 
Moreover, 
while the empirical competitive ratios of \ouralg are lower than ML, they have provable
competitive ratio which is scaled down by $\lambda$ compared to that of \experttwo.

\textbf{Training-testing distributional shifts.}
The above results consider that the training and testing instances
are drawn from the same distribution. Now, we consider an out-of-distribution (\emph{OOD}) testing case by adding perturbation noises to 30\% of the testing instances, and show the results in Table~\ref{table:result}. OOD
is commonly seen in practice, making ML predictions potentially untrusted.
Since the testing distribution shifts compared to the training distribution under the OOD setting, the performances of ML in terms of both average utility and competitive ratio
decrease and become worse than  \expert and \experttwo. Still, 
\expert and \experttwo outperform the other baselines (DMD, Greedy and Equal) in terms of the empirical competitive ratio, again showing their benefits in the worst-case competitive guarantee. The learning-augmented algorithm \ouralg improves the competitive ratio of ML with a large $\lambda$, showing its effects in providing the ML with guaranteed competitiveness. 

\begin{table}[!t]
\centering
\small
\begin{tabular}{l|l|c|c|c|c|c|c|c|c} 
\toprule
\multicolumn{1}{l}{} &           & \textbf{ML}     & 
\textbf{\expert} & \textbf{\experttwo} & \textbf{\ouralg-0.3} & \textbf{\ouralg-0.6}   & \textbf{DMD} & \textbf{Greedy} & \textbf{Equal}  \\ 
\hline
\multirow{2}{*}{\textbf{AVG}} & \textbf{In}    &  \textbf{0.9340} & 0.8959&0.9130&0.9311&0.9301&0.8715&0.8574&0.7246   \\ 
\cline{2-10}
                     & \textbf{OOD} & 0.8975 & 0.9036 & \textbf{0.9041} & 0.8953 & 0.8981 & 0.9016 & 0.9000 & 0.7528         \\ 
\hline
\multirow{2}{*}{\textbf{CR}}  & \textbf{In}    & \textbf{0.8645} & 0.8481 & 0.8565 & 0.8303 & 0.8223 & 0.8200 & 0.8048 & 0.5650       \\ 
\cline{2-10}
                     & \textbf{OOD} & 0.7916 & 0.8234 & \textbf{0.8411} & 0.7916 & 0.8003 & 0.8076 & 0.8048 & 0.5650   \\
\bottomrule
\end{tabular}
\caption{Comparison of average utility (AVG) and empirical competitive ratio (CR).
\ouralg-$n$ indicates
\ouralg with $\lambda=n$. 
``In'' and ``OOD'' mean in-distribution and out-of-distribution, respectively.
The average utility is normalized by that of OPT (i.e., 80.2771 and 78.8540 for the in-distribution and out-of-distribution  cases, respectively)....
Bold values represent the best for the respective metrics. }
\vspace{-2em}
\label{table:result}
\end{table}

\textbf{Performance under varying $\lambda$.}
Next, we show in Fig.~\ref{fig:avg_utility} the impact of $\lambda\in[0,1]$ on \ouralg in terms of the average utility. We see that under the in-distribution setting, when $\lambda$ increases,
the average utility of \ouralg can decrease due to
the increasingly more stringent worst-case robustness constraint \eqref{eqn:constraint_1_learning}. Interestingly,  \ouralg can achieve higher average utility than ML under some $\lambda$. This is due to the fact
that \experttwo used by \ouralg can correct the ML predictions
for some problem instances in which the original ML predictions do not perform well.
For the OOD setting, the average utility of ML is lower due to the distribution shift. By integrating \experttwo with ML, \ouralg is more beneficial in terms
of improving the average utility. 
This confirms our analysis of \ouralg in Theorems~\ref{thm:robustness} and ~\ref{thm:average}.

We show the empirical competitive ratios under different $\lambda$ in Fig.~\ref{fig:competitive ratio}.
In practice, it is difficult to evaluate the competitive ratio empirically since the adversarial samples for the algorithms under evaluation may not exist in the actual testing dataset under evaluation. 
As a result, a few unfavorable instances can affect the empirical
competitive ratio significantly.
Our results show that \ouralg has an empirical competitive ratio higher than the theoretical bound in Theorem~\ref{thm:robustness} (dotted line in Fig.~\ref{fig:competitive ratio}), which is also very common in practice.

Finally, we
show in Fig.~\ref{fig:violation_prob} the worst-case utility constraint violation probability 
for the pure standalone ML predictor. Naturally, when $\lambda$ increases, the worst-case utility constraint in \eqref{eqn:constraint_1_learning}  becomes tighter, making
the pure ML predictor violate the constraint more frequently.
This highlights the lack of worst-case utility guarantees of pure ML,
as well as the necessity of \ouralg to safeguard the ML predictions.
Thus, although ML empirically can have a good competitive
ratio (against OPT) as shown in Table~\ref{table:result} for
the in-distribution case, this empirical advantage is not always guaranteed. 
\begin{figure*}[!t]	
	\centering
 \subfigure[Average utility of \ouralg]{
		\label{fig:avg_utility}
		\includegraphics[width=0.295\textwidth]{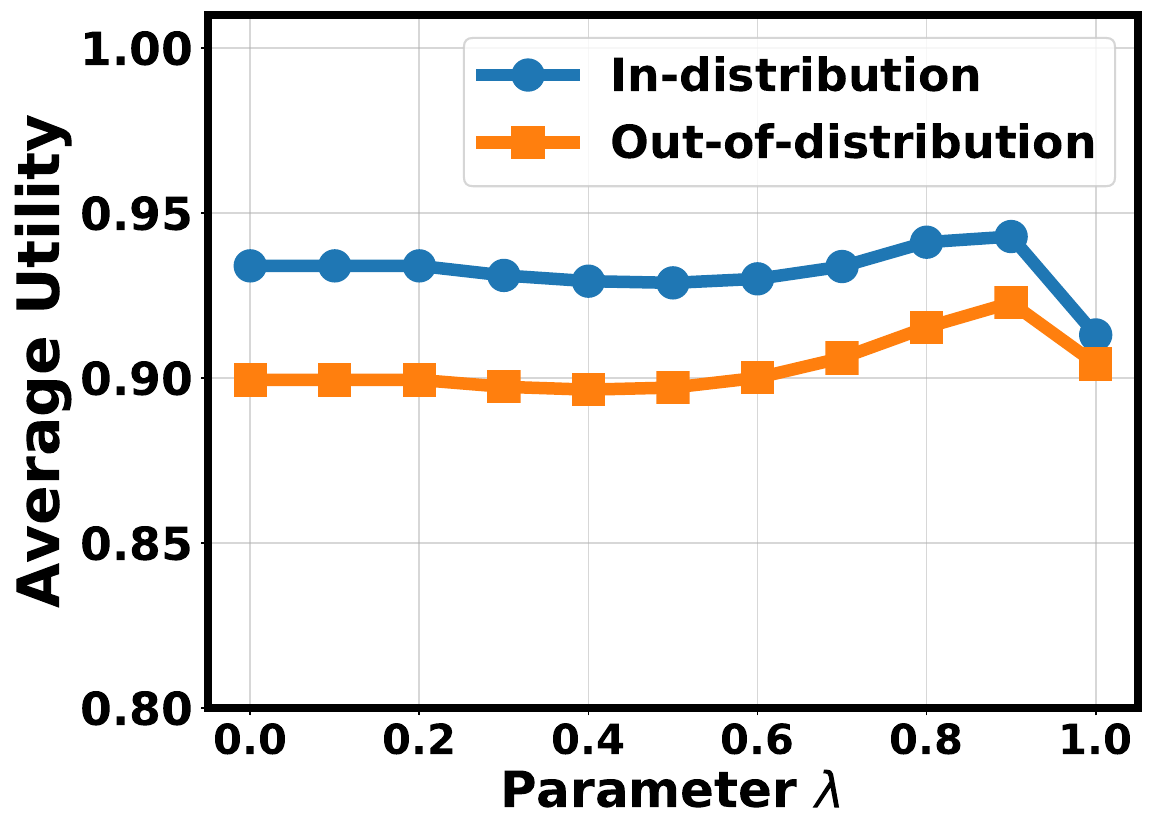}}
 \hspace{0.01\textwidth}
 \subfigure[Competitive ratio \ouralg]{
 \label{fig:cr}
 \label{fig:competitive ratio}
		 \includegraphics[width={0.295\textwidth}]{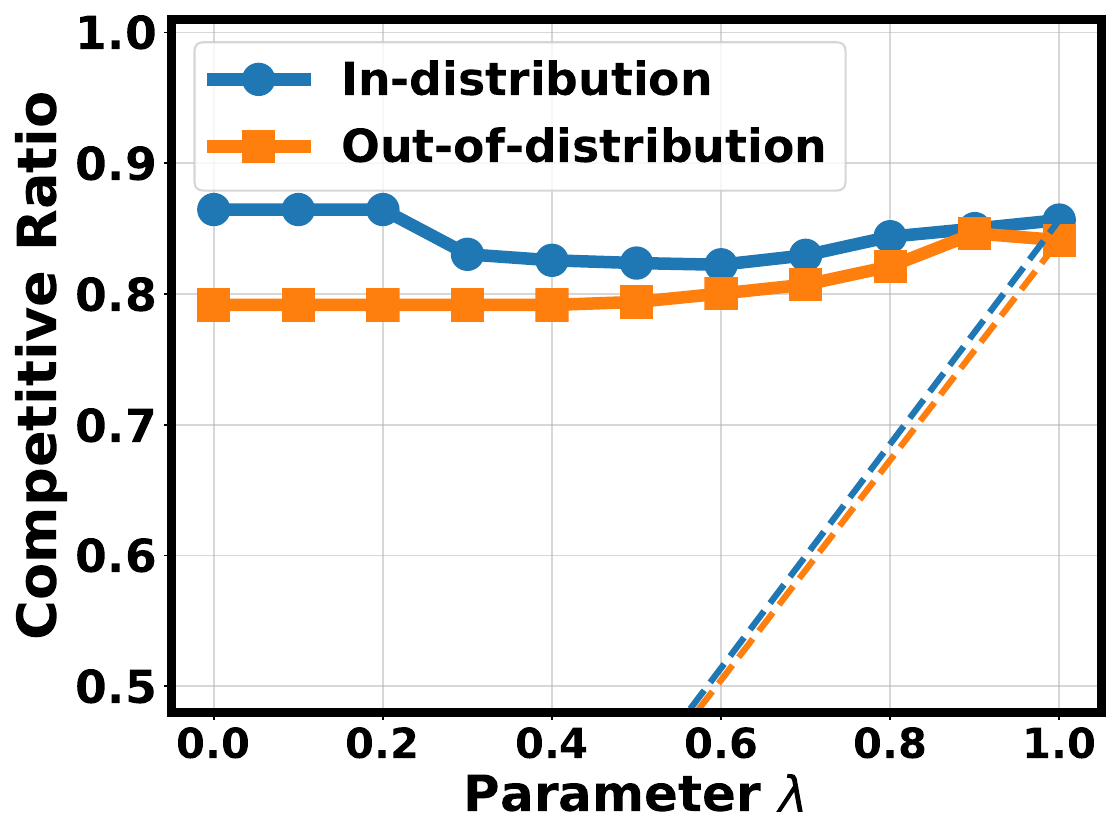}
   }
	\subfigure[Constraint violation rate of ML]{
		\label{fig:violation_prob}
		\includegraphics[width=0.295\textwidth]{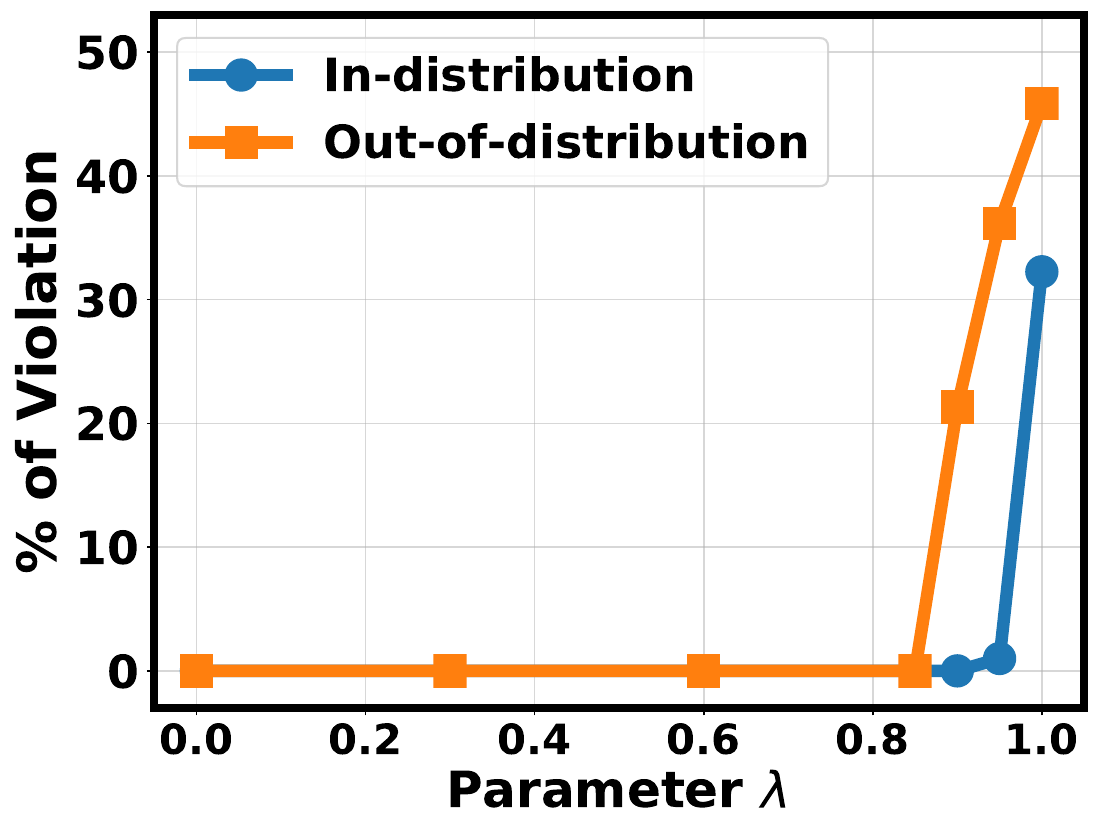}
	}%
	\centering
	\vspace{-0.4cm}	
	\caption{(a) Average utility of \ouralg with varying $\lambda\in[0,1]$; (b) Empirical competitive ratio of \ouralg with varying $\lambda\in[0,1]$
 (dotted lines represent the theoretical competitive ratio bounds);
 (c) Utility constraint \eqref{eqn:constraint_1_learning} violation  probability by the pure ML predictor.}
	\label{fig:details}
\end{figure*} 

\section{Related Works}\label{sec:related}

Online constrained allocation is a classic problem extensively studied in
the last few decades.
For example, 
 some earlier works \cite{OCO_devanur2009adwords,OCO_feldman2010online} solve online allocation by estimating a fixed Lagrangian multiplier using offline data,
 while
 other studies design online algorithms
by updating the Lagrangian multiplier or resource price in an online manner \cite{devanur2019near,OCO_agrawal2014fast,OCO_zinkevich2003online}. 
 Likewise, online algorithms have also been proposed for online stochastic optimization with distributional information \cite{Online_stochastic_optimization_non_stationary_jiang2020online}.
Online algorithms that allow budget violations are also available \cite{OnlineOpt_ConvexOptimization_TimeVaryingConstraints_Regret_LongboHuang_Performance_2021_LIU2021102240,Neely_Booklet,Neely_Universal}. 
In the context of network optimization, Lyapunov optimization can address 
various resource constraints by introducing resource queues (equivalent to the Lagrangian multiplier), but the extension to adversarial settings
with strict budget constraints is  challenging 
\cite{EnergyHarvesting_LearningAided_LyapunovOptimization_LongboHuang_TMC_2020_8807278,Neely_Booklet,Neely_LowComplexity_Lyapunov_JMLR_2020_JMLR:v21:16-494,Longbo_PowerOnlineLearningLyapuvno_Sigmetrics_2014_10.1145/2637364.2591990}.

Online allocation with budget constraints in adversarial settings is very challenging and has not been fully resolved yet.
Concretely, for online allocation with inventory constraints,
competitive online algorithms are designed by pursuing a
pseudo-optimal algorithm, but the utility function either takes a single scalar 
\cite{OnlineAllocation_SingleInventory_MinghuaChen_Sigmetrics_2019_10.1145/3322205.3311081}
or is separable over multiple dimensions \cite{OnlineAllocation_MultiInventory_MinghuaChen_Sigmetrics_2022_10.1145/3530902}. A recent study  \cite{OnlineAllocation_Dual_mirror_descent_Google_ICML_2020_balseiro2020dual} considers online allocation with a more general convex utility function and proposes dual mirror descent (DMD)
to update the Lagrangian multiplier given stochastic inputs at each round,
while the extension to adversarial settings has been considered more recently in \cite{OnlineAllocation_DualMirroDescent_Google_OperationalResearch_2022_doi:10.1287/opre.2021.2242} and extension
to uncertain time horizons is studied in \cite{OnlineAllocation_UncertainHorizon_DMD_Sigmetrics_2023_10.1145/3578338.3593559}.
 Nonetheless, these studies do not apply to online budget replenishment,
which we address by proposing provably-competitive \expert and \experttwo.

ML predictors/policies have been emerging for exploiting
the distributional information of problem inputs and hence improving the average performance
of various (online) optimization problems \cite{L2O_LearningToOptimize_Berkeley_ICLR_2017,L2O_Survey_Benchmark_ZhangyangWang_WotaoYin_arXiv_2021_chen2021learning,Shaolei_LearningRobustCombinatorial_Zhihui_Infocom_2022}. 
For example, online scheduling, resource management, and classic secretary problems
\cite{L2O_NewDog_OldTrick_Google_ICLR_2019,L2O_Combinatorial_Reinforcement_AAAI_2020,L2O_Scheduling_DRL_Infocom_2019_8737649,Shaolei_LearningRobustCombinatorial_Zhihui_Infocom_2022} have all been considered. 
Nonetheless,
 a major drawback of these standalone ML-based optimizers is
that they do not have worst-case performance guarantees
and may have very high or even unbounded losses in the worst case.
As a consequence, they may not be suitable for mission-critical applications. 
While constrained ML-based policies \cite{constrained_MDP_efroni2020exploration,constrained_RL_primal_dual_ding2020natural,conservative_bandits_wu2016conservative,Conservative_Bandits_LinearContextual_NIPS2017_bdc4626a} are available, they focus on orthogonal
challenges (i.e., unknown cost/utility functions) and typically
focus on the average constraint, rather than worst-case utility constraint
for any problem instance.

\ouralg is relevant to the emerging field
of learning-augmented algorithms 
\cite{Shaolei_SOCO_RobustLearning_OnlineOpt_MemoryCosts_Infocom_2023,SOCO_ML_ChasingConvexBodiesFunction_Adam_COLT_2022,OnlineOpt_ML_Adivce_Survey_2016_10.1145/2993749.2993766,OnlineOpt_Learning_Augmented_RobustnessConsistency_NIPS_2020,OnlineOpt_ML_Adivce_Survey_2016_10.1145/2993749.2993766,OnlineOpt_ML_Advice_CompetitiveCache_Google_JACM_2021_10.1145/3447579}.
The goal of typical learning-augmented algorithms is to improve
the worst-case competitive ratio when the ML prediction is perfect,
while bounding the worst-case competitive ratio when ML predicition
is arbitrarily bad. While it has been considered
in a variety of settings, a learning-augmented algorithm for online allocation with replenishable
budgets is still lacking. Thus, \ouralg addresses this gap and is
the first learning-augmented algorithm for online allocation with replenishable budgets that
offers worst-case utility guarantees for any problem instance.

\section{Conclusion}

In this paper, we study online resource allocation with replenishable budgets,
and propose
novel competitive algorithms, called 
\expert and \experttwo, that 
conservatively adjusts dual variables while opportunistically
utilizing available resources.
We prove, for the first time, that \expert and \experttwo both achieve bounded asymptotic
competitive ratios in adversarial settings as the number of decision rounds
$T\to\infty$. 
In particular, under the mild assumption
that the budget is replenished every $T^*$ rounds,
\ouralgtwo can improve the asymptotic competitive ratio over
\expert.
 Then, to
address the conservativeness of \expert,
we move beyond the worst-case and propose
\ouralg, a novel learning-augmented algorithm
for our problem setting.
\ouralg can provably improve
the average utility compared to \expert and \experttwo when the ML predictor
is properly trained, while still offering worst-case utility guarantees.
Finally, we perform simulation studies using online power allocation 
with energy harvesting. Our results validate our analysis
and demonstrate the empirical benefits of \ouralg compared to existing baselines.

\section*{ACKNOWLEDGEMENT}
This work was supported in part by the NSF under grants CNS-1910208 and CNS-2007115.


\received{February 2023}
\received[revised]{January 2024}
\received[accepted]{January 2024}

\appendix

\section*{Appendix}
\section{Proof of Theorem~\ref{thm:expertbound}}\label{appendix:proof_expert}
We now prove Theorem~\ref{thm:expertbound}
and first restate the convergence lemma of online mirror descent.
\begin{lemma}[\cite{OnlineAllocation_DualMirroDescent_Google_OperationalResearch_2022_doi:10.1287/opre.2021.2242,tutorial_online_learning_orabona2019modern}]\label{lma:mirrordescent}
Let $V_h(x,y)=h(x)-h(y)-\triangledown h(y)^{\top}(x-y)$ be the Bregman divergence based on a $\sigma$-strongly convex function $h$.  If $w_t(\mu)$ is a convex function with respect to $\mu\in\mathcal{D}$ where $\mathcal{D}$ is a convex set and its sub-gradient satisfies $\|\partial_{\mu} w_t(\mu)\|_{\infty}\leq G$, by updating the variable $\mu_{t+1}=\arg\min_{\mu\in\mathcal{D}} \mu^\top\partial_{\mu} w_t(\mu)+\frac{1}{\eta}V_h(\mu,\mu_t)$ from some initial variable $\mu_1$, it holds for any $\mu\in\mathcal{D}$ that
\begin{equation}
\sum_{t=1}^Tw_t(\mu_t)-w_t(\mu)\leq \frac{G^2\eta}{2\sigma}T+\frac{1}{\eta}V_h(\mu,\mu_1).
\end{equation}
\end{lemma}

\noindent\textbf{Proof of Theorem \ref{thm:expertbound}}

\begin{proof}
We define $\mathcal{T}_A=\left\{\tau_1,\cdots, \tau_{|\mathcal{T}_A|}\right\}$ as a set of rounds when  $\hat{x}_t$ violates the budget constraint, i.e. $\forall \tau\in \mathcal{T}_A$, there exists a dimension $m$ such that $(\hat{x}_{\tau})_m>(B_{\tau}+E_{\tau})_m$.
By our algorithm design, if $t\in \mathcal{T}_A$, we choose $x_t=0$ and $g_t=0$.
Define a sequence of functions as
\begin{equation}
w_t(\mu)=\mu_t^\top g_t=\left\{\begin{matrix}
\mu_t^\top(\rho-\hat{x}_t), & t\notin \mathcal{T}_A,\\ 
0, & t\in \mathcal{T}_A.
\end{matrix}\right. 
\end{equation}
By Lemma \ref{lma:mirrordescent}, we have
\begin{equation}\label{eqn:dual_convergence}
\sum_{t=1}^T w_t(\mu_t)-w_t(\mu)\leq \frac{G^2\eta}{2\sigma}T +\frac{1}{\eta}V_h(\mu,\mu_1),
\end{equation}
where $G=\sup \|g_t\|_{\infty}\leq \bar{\rho}+\|\bar{x}\|_{\infty}$.
By our algorithm design, $\forall t\notin \mathcal{T}_A$, the action is chosen as $x_t=\arg\max_{x\in\mathcal{X}}\{f_t(x)-\mu_t^\top x\}$, we have
$f_t(x_t^*)\leq f_t(x_t)+\mu_t^\top (x_t^*-x_t)$ and $0=f_t(0)\leq f_t(x_t)-\mu_t^\top x_t$.
Thus we have 
\begin{equation}\label{eqn:adversarial_utility_bound}
\begin{split}
\alpha f_t(x_t)&=f_t(x_t)+(\alpha-1)f_t(x_t)\\
&\geq f_t(x_t^*)-\mu_t^\top x_t^* +\mu_t^\top x_t+(\alpha-1)f_t(x_t)\\
&\geq f_t(x_t^*)-\mu_t^\top x_t^* +\mu_t^\top x_t+(\alpha-1)\mu_t^\top x_t\\
&=f_t(x_t^*)-\alpha \mu_t^\top (\rho-x_t) - \mu_t^
\top x_t^* +\alpha \mu_t^\top\rho\\
&\geq f_t(x_t^*)-\alpha w_t(\mu_t),
\end{split}
\end{equation}
where the last inequality holds by setting $\alpha= \max_{m\in[M]}  
\frac{\bar{x}_m}{\rho_m}$.

Then for any $\mu>0$, we have
\begin{equation}\label{eqn:regret}
\begin{split}
&OPT(y)-\alpha F_T(y)\\
\leq &\sum_{t=1}^T f_t(x_t^*)-\alpha \sum_{t\notin\mathcal{T}_A}f_t(x_t)\\
\leq & \sum_{t=1}^T f_t(x_t^*)-\sum_{t\notin\mathcal{T}_A}f_t(x_t^*)+\sum_{t\notin\mathcal{T}_A}\alpha w_t(\mu_t)\\
\leq  & \sum_{t\in\mathcal{T}_A}f_t(x_t^*)+\alpha\sum_{t\notin\mathcal{T}_A}
w_t(\mu)+\alpha\left(\frac{G^2\eta}{2\sigma}T +\frac{1}{\eta}V_h(\mu,\mu_1)\right)\\
\leq &|\mathcal{T}_A|\bar{f}+\alpha\sum_{t\notin\mathcal{T}_A}
\mu^\top(\rho-x_t)+\alpha\left(\frac{G^2\eta}{2\sigma}T +\frac{1}{\eta}V_h(\mu,\mu_1)\right)
\end{split}
\end{equation}
where the first inequality holds because the utility are non-negative, the second inequality holds by \eqref{eqn:adversarial_utility_bound}, the third inequality holds by Lemma \ref{lma:mirrordescent}, and the last inequality holds by $f_t\leq \bar{f}$.

Now it remains to choose $\mu$ to get the bound. 
If $|\mathcal{T}_A|=0$, set $\mu=0$, and the bound holds.
Otherwise, we choose $\mu$ as follows. 
Define $\mathcal{M}_A$ is the set of resources of which the corresponding constraints are violated, i.e.  for $m\in\mathcal{M}_A$, $\exists t\in[T]$ such that $\hat{x}_{m,t}>(B_t+E_t)_m$.
Since the consumed resource plus $\hat{x}_{t,m}$ is larger than the initial budget $B_{1,m}$ when the constraint resource $m$ is violated and $\hat{x}_{t,m}\leq \bar{x}_m$, it holds for resource $m\in \mathcal{M}_A$ that 
\begin{equation}\label{eqn:remaining_budget}
\begin{split}
\sum_{t\notin\mathcal{T}_A}x_{t,m}+\bar{x}_m\geq B_{1,m}=\rho_m T.
\end{split}
\end{equation}
We choose one resource $j\in \mathcal{M}_A$ and set $\mu=\frac{\bar{f}}{\alpha\rho_j}e_j$ where $e_j$ is a unit vector with $j$th entry being one and other entries being zero, it holds that 
\begin{equation}\label{eqn:bound_w_mu}
\begin{split}
&\alpha\sum_{t\notin\mathcal{T}_A}
\mu^\top(\rho-x_t)\\
=&\alpha\sum_{t\notin\mathcal{T}_A}
\mu_j(\rho_j-x_{t,j})\\
\leq& \alpha
(T-|\mathcal{T}_A|)\mu_j\rho_j-\alpha \mu_j(T\rho_j-\bar{x}_j)\\
\leq & -\alpha
|\mathcal{T}_A|\mu_j\rho_j+\alpha \mu_j\bar{x}_j\\
\leq &-|\mathcal{T}_A|\bar{f}+\alpha\bar{f},
\end{split}
\end{equation}
where the first inequality holds by \eqref{eqn:remaining_budget}, and the last inequality holds by the choice of $\mu$.

Substituting \eqref{eqn:bound_w_mu} into \eqref{eqn:regret}, we get the bound as
\begin{equation}
\begin{split}
OPT(y)-\alpha R_T^{DMD}(y)
\leq \alpha\bar{f}+\frac{\alpha G^2\eta T}{2\sigma}+\frac{\alpha}{\eta}V_h(\mu,\mu_1),
\end{split}
\end{equation}
where $\alpha= \sup_{m\in[M]} 
\frac{\bar{x}_m}{\rho_m}$, and $\mu=0$ if $\mathcal{M}_A=\emptyset$. Otherwise, $\mu=\frac{\bar{f}}{\alpha\rho_j}e_j, j=\arg\min_m V_h(\frac{\bar{f}}{\alpha\rho_m}e_m,\mu_1), m\in\mathcal{M}_A\}$.
Thus, we complete the proof.
\end{proof}

\section{Proof of Theorem~\ref{thm:enhanced_expertbound_cr}}\label{sec:prooftheorem3.2}
\begin{figure*}[!t]	
	\centering
		 \includegraphics[width={0.8\textwidth}]{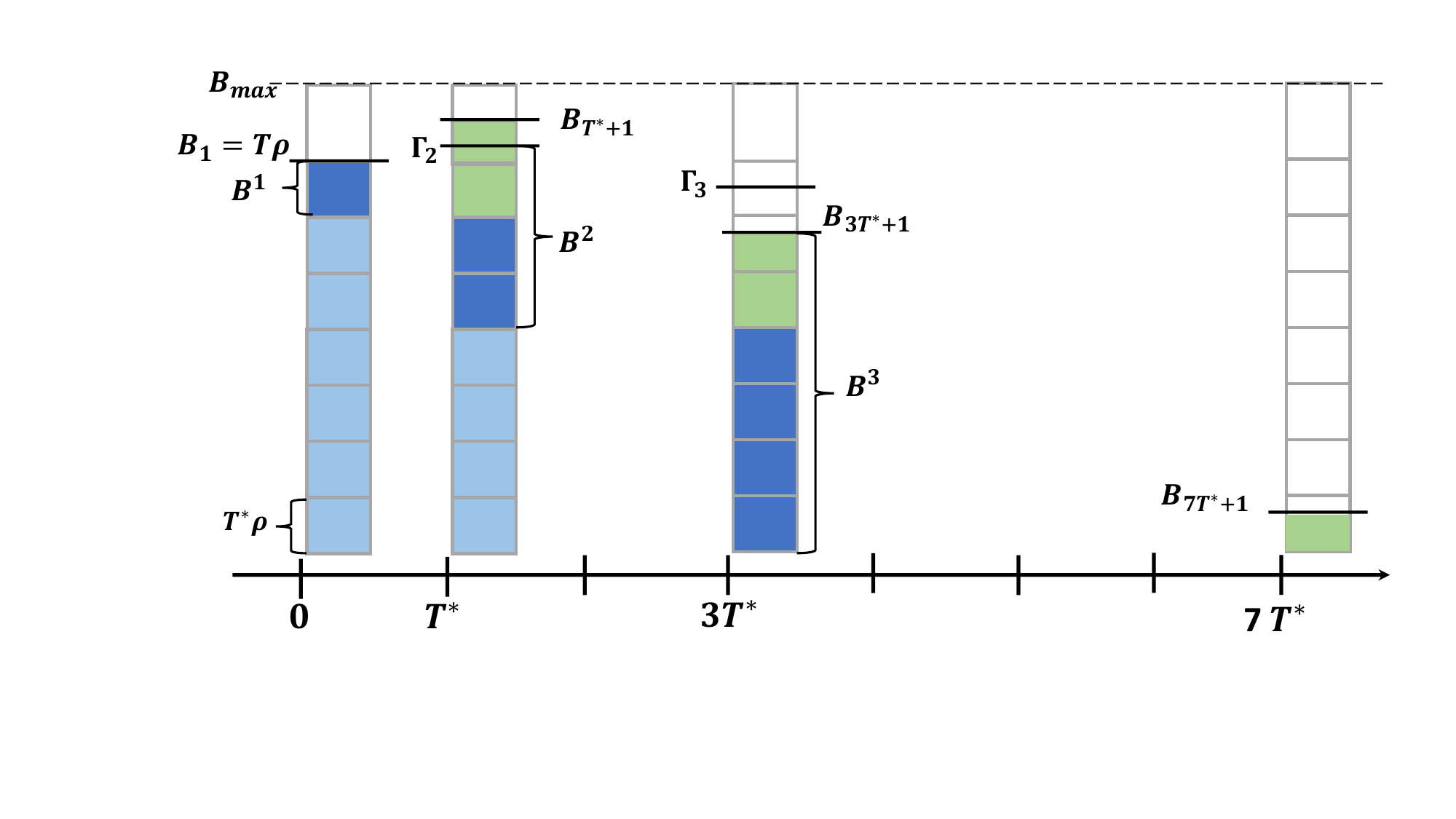}
	\caption{An example of budget assignment with $T=7T^*$. Colored rectangles indicate the amount of remained budget and white rectangles are the spaces in the storage. Dark blue rectangles indicate permanent budgets $2^{i-1}T^*\rho$ for the current frame. Light blue rectangles indicate permanent budgets for the future frames $(T-(2^{(i)}-1)T^*)\rho$. Green rectangles indicate the budget accumulation  $\min\{ B_{T_{i-1}+1}-(T-(2^{i-1}-1)T^*)\rho,  2^{i-2}T^*\rho_{\max}\odot\beta\}$.   }
	\label{fig:budget_illustrtion}
\end{figure*}

\begin{lemma}\label{lma:assignbudgetlowerbound}
If a fixed budget $B^{(i)}=2^{i-1}T^*\rho+\Omega_{i}$ where $\Omega_{i}=\min\{B_{T_{i-1}+1}-(T-(2^{i-1}-1)T^*)\rho, 2^{i-2}T^*\rho_{\max}\odot\beta\}$ where $\rho_{\max}=B_{\max}/T$ is assigned to each frame $i,1\leq i\leq K$ with $2^{i-1}T^*$ rounds, the additive budget $\Omega_i$ is greater or equal to equivalent additive budget $\hat{\Omega}_i,1\leq i\leq K-1$ which is expressed as 
\begin{subequations}
\label{eqn:effectivebudget}
\begin{align}
&\hat{\Omega}_1=0\\
&\hat{\Omega}_2=\min\{T\rho_{\max}-T\rho,E'_{\min}\}\\
&\hat{\Omega}_{i}=\min\{T\rho_{\max}-2^{i-3}T^*\rho_{\max}\odot\beta-(T-(2^{i-2}-1)T^*)\rho,2^{i-2}E'_{\min}\},3\leq i\leq K
\end{align}
\end{subequations}
where $E_{\min}'=\min\{E_{\min}, T^*\rho_{\max}\odot\beta\}$.
\end{lemma}
\begin{proof}  
We prove that the equivalent additive budget $\hat{\Omega}_{i}$ does not exceed the true additive budget $\Omega_{i}$ for any frame $i$.

For the first frame, it is obvious that $\hat{\Omega}_1\leq \Omega_1=0$ holds. 
For the second frame, we discuss the value of $\Omega_2$ in the following cases. 

Firstly, if for a resource $m\in[M]$, $B_{T^*+1,m}-(T-T^*)\rho_m\leq T^*\beta_m\rho_{\max,m}$, the additive budget $\Omega_{2,m}$ is $B_{T^*+1,m}-(T-T^*)\rho_m$, and it comprises the replenishment in the first frame $\sum_{t=1}^{T_1}E_{t,m}$ and the unconsumed budget in the first frame $B^1_m-\sum_{t=1}^{T_1}x_{t,m}$. We can bound the replenishment in the first frame as 
\begin{equation}\label{eqn:replenishlowerbound}
\sum_{t=1}^{T_1}E_{t,m}\geq \min\{B_{\max,m}-T\rho_m,E_{\min,m}\}\geq \min\{B_{\max,m}-T\rho,E_{\min,m}'\}=\hat{\Omega}_{2,m}.
\end{equation}
The reason is that if the truly replenished budget of resource $m$ at each round of the first frame is not constrained by $B_{\max,m}$, i.e. $E_{t,m}=\hat{E}_{t,m},\forall t\in[1,T_1]$, we have $\sum_{t=1}^{T_1}E_{t,m}=\sum_{t=1}^{T_1}\hat{E}_{t,m}\geq E_{\min,m}\geq E_{\min,m}'$. Otherwise, we must have $\sum_{t=1}^{T_1}E_{t,m}\geq B_{\max,m}-T\rho_m$ since $B_{\max,m}-T\rho_m$ is the minimum replenished budget such that the replenishment is constrained by the budget cap $B_{\max,m}$. Therefore for the first case, we always have for the resource $m$, $\hat{\Omega}_{2,m}\leq \sum_{t=1}^{T_1}E_{t,m}\leq \Omega_{2,m}$. 

For the second case when $B_{T^*+1,m}-(T-T^*)\rho_m> T^*\beta_m\rho_{\max,m}$ for resource $m$, we have $\Omega_{2,m}= T^*\beta_m\rho_{\max,m}$. Thus, we still have $\hat{\Omega}_{2,m}\leq E'_{\min,m}\leq T^*\beta_m\rho_{\max,m}=  \Omega_{2,m}$.

Since the inequality holds for all the resources $m$, we have $\hat{\Omega}_{2}\leq\Omega_{2}$.

For the $i$th ($3\leq i\leq K$) frame, we discuss for the value of $\Omega_i$ in the following cases.

Firstly, if for a resource $m$, $B_{T_{i-1}+1,m}-(T-(2^{i-1}-1)T^*)\rho_m\leq  2^{i-2}T^*\beta_m\rho_{\max,m}$, then the additive budget $\Omega_{i,m}$ includes the replenishment in the $(i-1)$th frame $\sum_{t=T_{i-2}+1}^{T_{i-1}}E_{t,m}$, the unconsumed assigned budget in the $(i-1)$th frame $B^{i-1}_m-\sum_{t=T_{i-2}+1}^{T_{i-1}}x_{t,m}$, and the possibly saved budget $[B_{T_{i-2}+1,m}-(T-(2^{i-2}-1)T^*)\rho_m - 2^{i-3}T^*\beta_m\rho_{\max,m}]^+$ at the beginning of $(i-1)$th frame.  The truly replenished budget in the $(i-1)$th frame can be bounded as 
\begin{equation}\label{eqn:boundreplenishcap}
\sum_{t=T_{i-2}+1}^{T_{i-1}}E_{t,m}\geq \min\{B_{\max,m}-B_{T_{i-2}+1,m},2^{i-2}E'_{\min,m}\}.
\end{equation}
The reason is that if the replenishment at each round of the $(i-1)$th frame is not constrained by $B_{\max,m}$, i.e. $E_{t,m}=\hat{E}_{t,m},\forall t\in[T_{i-2}+1,T_{i-1}]$, we have $\sum_{t=T_{i-2}+1}^{T_{i-1}}E_{t,m}=\sum_{t=T_{i-2}+1}^{T_{i-1}}\hat{E}_{t,m}\geq 2^{i-2} E_{\min,m}\geq 2^{i-2} E_{\min,m}'$. Otherwise, we must have $\sum_{t=T_{i-2}+1}^{T_{i-1}}E_{t,m}\geq B_{\max,m}-B_{T_{i-2}+1,m}$ since $B_{\max,m}-B_{T_{i-2}+1,m}$ is the minimum replenished budget such that the replenishment is constrained by the budget cap $B_{\max,m}$.

If it holds at the beginning of the $(i-1)$th frame that $B_{T_{i-2}+1,m}\leq  (T-(2^{i-2}-1)T^*)\rho_m + 2^{i-3}T^*\beta_m\rho_{\max,m}$, we further have
\begin{equation}
\Omega_{i,m}\geq \sum_{t=T_{i-2}+1}^{T_{i-1}}E_{t,m}\geq \min\{B_{\max,m}-2^{i-3}T^*\beta_m\rho_{\max,m}-(T-(2^{i-2}-1)T^*)\rho_m,2^{i-2}E'_{\min,m}\}=\hat{\Omega}_{i,m}.
\end{equation}
 Otherwise,  $[B_{T_{i-2}+1,m}-(T-(2^{i-2}-1)T^*)\rho_m - 2^{i-3}T^*\beta_m\rho_{\max,m}]^+$ is positive and is included in $\Omega_{i,m}$. Under such a case, we have
\begin{equation}
\begin{split}
\hat{\Omega}_{i,m}=&\min\{(T-2^{i-3}T^*\beta_m)\rho_{\max,m}-(T-(2^{i-2}-1)T^*)\rho_m,2^{i-2}E'_{\min,m}\}\\
\leq &\min\{B_{\max,m}-B_{T_{i-2}+1,m},2^{i-2}E'_{\min,m}\}+B_{T_{i-2}+1,m}-(T-(2^{i-2}-1)T^*)\rho_m - 2^{i-3}T^*\beta_m \rho_{\max,m}\\
\leq & \sum_{t=T_{i-2}+1}^{T_{i-1}}E_{t,m}+B_{T_{i-2}+1,m}-(T-(2^{i-2}-1)T^*)\rho_m - 2^{i-3}T^*\beta_m \rho_{\max,m}\leq \Omega_{i,m},
\end{split}
\end{equation}
where the first inequality holds because $\min\{A+B,C\}\leq \min\{A,C\}+B$ for $A,B,C\geq 0$, the second inequality holds by \eqref{eqn:boundreplenishcap}, and the last inequality holds since $\sum_{t=T_{i-2}+1}^{T_{i-1}}E_{t,m}$ and $[B_{T_{i-2}+1,m}-(T-(2^{i-2}-1)T^*)\rho_m - 2^{i-3}T^*\beta_m\rho_{\max,m}]^+$ are both included in $\Omega_{i,m}$.

Secondly, if $B_{T_{i-1}+1,m}-(T-(2^{i-1}-1)T^*)\rho_m>  2^{i-2}T^*\beta_m\rho_{\max,m}$, the additive budget $\Omega_{i,m}=2^{i-2}T^*\beta_m\rho_{\max,m}$, and we have $\hat{\Omega}_{i,m}\leq 2^{i-2}E'_{\min,m}\leq 2^{i-2}T^*\beta_m\rho_{\max,m}=\Omega_{i,m}$.

Since the inequality holds for all the resources $m$, we have $\hat{\Omega}_{i}\leq\Omega_{i}$ for $3\leq i\leq K$.
\end{proof}

\textbf{Proof of Theorem \ref{thm:enhanced_expertbound_cr}}
\begin{proof}
Since dual mirror descent is applied to each frame, using similar techniques as the proof of Theorem \ref{thm:expertbound}, we can prove that within each frame $i, i\in[K]$, given the choice of $\eta$ and $\mu$, it holds that
\begin{equation}\label{eqn:frame_regret}
\sum_{t=T_{i-1}}^{T_i}f_t(x_t^*)-\alpha_i f_t(x_t)\leq \alpha_i \bar{f}+\alpha_i(\bar{\rho}^{(i)}+\|\bar{x}\|_{\infty})\sqrt{\frac{V_h(\mu,\mu_1)(2^{i-1}T^*)}{2\sigma}},
\end{equation}
where $x_t^*$ is the offline-optimal solution for the whole episode with length $T$, $\alpha_i=\sup_{m\in[M]}\frac{\bar{x}_m}{\rho_m^{(i)}}$, and $\bar{\rho}^{(i)}=\sup_{m\in[M]}\rho^{(i)}_m$.

To use the doubling trick, we need to bound $\rho^{(i)}=\frac{B^{(i)}}{2^{i-1}T^*}$. By Lemma \ref{lma:assignbudgetlowerbound},  we have $\rho^1=\rho$, $\rho^2=\frac{2T^*\rho+\Omega_{i}}{2 T^*}\geq \frac{2T^*\rho+\min\{T\rho_{\max}-T\rho,E'_{\min}\}}{2 T^*}=\rho +\min\{\frac{T\rho_{\max}-T\rho}{2T^*},\frac{\rho_{\max}\odot\beta}{2}, \frac{E_{\min}}{2T^*}\}$, and for $3\leq i\leq K$, we have
\begin{equation}
\begin{split}
\rho^{(i)}=&B^{(i)}/(2^{i-1} T^*)=\frac{2^{i-1}T^*\rho+\Omega_{i}}{2^{i-1} T^*}\geq \frac{2^{i-1}T^*\rho+\hat{\Omega}_{i}}{2^{i-1} T^*}\\
= & \rho+\frac{\min\left\{T\rho_{\max}-2^{i-3}T^*\rho_{\max}\odot\beta-(T-(2^{i-2}-1)T^*)\rho,2^{i-2}E'_{\min}\right\}}{2^{i-1} T^*}\\
= &\rho +\min\left\{\frac{1}{2^{i-1}}\left(\frac{T\rho_{\max}}{T^*}-\frac{T+T^*}{T^*}\rho\right)+\frac{\rho}{2}-\frac{\rho_{\max}\odot\beta}{4},\frac{\rho_{\max}\odot\beta}{2}, \frac{E_{\min}}{2T^*}\right\},
\end{split}
\end{equation}
where the first inequality holds since $\Omega_i\geq \hat{\Omega}_i$, and the last equality holds since $E_{\min}'=\min\{E_{\min},\\ T^*\rho_{\max}\odot\beta\}$. 
If it holds for a resource $m$ that $B_{\max,m}<(T+T^*)\rho_m$, we have $\frac{1}{2^{i-1}}\left(\frac{T\rho_{\max,m}}{T^*}-\frac{T+T^*}{T^*}\rho_m\right)+\frac{\rho_m}{2}-\frac{\beta_m\rho_{\max,m}}{4}\geq \frac{T\rho_{\max,m}-(T+T^*)\rho_m}{4T^*}+\frac{\rho_m}{2}-\frac{\beta\rho_{\max,m}}{4}$. By optimally choosing $\beta_m=\frac{T}{3T^*}-\frac{T-T^*}{3T^*}\frac{\rho_m}{\rho_{\max,m}}$, we have 
\begin{equation}
\begin{split}
\rho^{(i)}_m\geq & \rho_m+\min\left\{\frac{T\rho_{\max,m}-(T+T^*)\rho_m}{4T^*}+\frac{\rho_m}{2}-\frac{\beta_m\rho_{\max,m}}{4},\frac{\beta_m \rho_{\max,m}}{2}, \frac{E_{\min,m}}{2T^*} \right\}\\
=&\rho_m+\min \left\{\frac{T\rho_{\max,m}}{6T^*}-\frac{(T-T^*)\rho_m}{6T^*},\frac{E_{\min,m}}{2T^*}\right\}.
\end{split}
\end{equation}
If it holds for a resource $m$ that $B_{\max,m}\geq (T+T^*)\rho_m$, we have $\frac{1}{2^{i-1}}\left(\frac{T\rho_{\max,m}}{T^*}-\frac{T+T^*}{T^*}\rho_m\right)+\frac{\rho_m}{2}-\frac{\beta_m\rho_{\max,m}}{4}\geq \frac{1}{2^{K-1}}\left(\frac{T\rho_{\max,m}}{T^*}-\frac{T+T^*}{T^*}\rho_m\right)+\frac{\rho_m}{2}-\frac{\beta_m\rho_{\max,m}}{4}\geq \frac{T\rho_{\max,m}-(T+T^*)\rho_m}{T+T^*}+\frac{\rho_m}{2}-\frac{\beta_m\rho_{\max,m}}{4}$ given that $T\geq (2^{K-1}-1)T^*$. By optimally choosing $\beta_m=\frac{4T}{3(T+T^*)}-\frac{2\rho_m}{3\rho_{\max,m}}$, we have
\begin{equation}
\begin{split}
\rho^{(i)}_m\geq &\rho_m +\min\left\{\frac{T\rho_{\max,m}-(T+T^*)\rho_m}{T+T^*}+\frac{\rho_m}{2}-\frac{\beta_m\rho_{\max,m}}{4},\frac{\beta_m \rho_{\max,m}}{2}, \frac{E_{\min,m}}{2T^*}\right\}\\
=&\rho_m +\min\left\{\frac{2T\rho_{\max,m}}{3(T+T^*)}-\frac{\rho_m}{3},\frac{E_{\min,m}}{2T^*}\right\}.
\end{split}
\end{equation}
Therefore, we can bound $\rho^{(i)}_m$ as $\rho^{(i)}_m\geq \rho_m +\min\left\{\frac{2T\rho_{\max,m}}{3(T+T^*)}-\frac{\rho_m}{3}, \frac{E_{\min,m}}{2T^*}\right\}$ when $B_{\max,m}\geq (T+T^*)\rho_m$ and $\rho^{(i)}_m\geq \rho_m +\min\left\{\frac{T\rho_{\max,m}}{6T^*}-\frac{(T-T^*)\rho_m}{6T^*}, \frac{E_{\min,m}}{2T^*}\right\}$ when $B_{\max,m}<(T+T^*)\rho_m$.  We define $\Delta \rho_m=\min\{\frac{2T\rho_{\max,m}}{3(T+T^*)}-\frac{\rho_m}{3},  \frac{E_{\min,m}}{2T^*}\}$ when $B_{\max,m}\geq (T+T^*)\rho_m$ and $\Delta \rho_m=\min\left\{\frac{T\rho_{\max,m}}{6T^*}-\frac{(T-T^*)\rho_m}{6T^*}, \frac{E_{\min,m}}{2T^*}\right\}$ when $B_{\max,m}< (T+T^*)\rho_m$. Thus, we have $\rho^{(i)}_m\geq \rho_m + \Delta\rho_m$.

Also, we can get the upper bound of $\rho^{(i)}$ for $i\in[2,K]$ as
$
\rho^{(i)}\leq \frac{2^{i-1}T^*\rho+2^{i-2}T^*\rho_{\max}\odot\beta}{2^{i-1} T^*}=\rho+\frac{\rho_{\max}\odot\beta}{2},
$ where the inequality holds because $\Omega_i\leq 2^{i-2}T^*(\beta\odot\rho_{\max})$.
Thus we have $\bar{\rho}^{(i)}=\sup_{m\in[M]}\rho^{(i)}_m\leq \bar{\rho}+\frac{\bar{\beta}}{2}\bar{\rho}_{\max},$
where $\bar{\rho}_{\max}=\max_m\rho_{\max,m}$, $\bar{\beta}=\max_m\beta_{m}$. When $B_{\max,m}\geq (T+T^*)\rho_m$, the optimal $\rho_m\leq \frac{4}{3}$ as $T\rightarrow \infty$. When $B_{\max,m}< (T+T^*)\rho_m$, the optimal $\rho_m\leq \frac{2}{3}$ as $T\rightarrow \infty$ since $\frac{T}{T+T^*}<\frac{\rho_m}{\rho_{\max,m}}\leq 1$.

Define $\hat{\alpha}=\min_{m\in[M]}\frac{\bar{x}_m}{\rho_m+\Delta\rho_m}$. By summing up frames with the lower and upper bounds of $\rho^{(i)}$, we get
\begin{equation}
\begin{split}
\sum_{t=1}^{T}&f_t(x_t^*)-\hat{\alpha} f_t(x_t)\leq \sum_{t=1}^{3T^*}f_t(x_t^*)-\hat{\alpha} f_t(x_t)+\sum_{i=3}^K\sum_{t=T_{i-1}}^{T_i}f_t(x_t^*)-\alpha_i f_t(x_t)\\
&\leq 3\bar{f}T^*+\sum_{i=3}^K \alpha_i \bar{f}+\alpha_i(\bar{\rho}+\frac{\bar{\beta}}{2}\bar{\rho}_{\max}+\|\bar{x}\|_{\infty})\sqrt{\frac{V_h(\mu,\mu_1)(2^{i-1}T^*)}{2\sigma}}\\
&\leq 3\bar{f}T^*+\hat{\alpha}K\bar{f}+\hat{\alpha}(\bar{\rho}+\frac{\bar{\beta}}{2}\bar{\rho}_{\max}+\|\bar{x}\|_{\infty})\sqrt{\frac{V_h(\mu,\mu_1)}{2\sigma}}\sum_{i=3}^K\sqrt{(2^{i-1}T^*)}\\
&\leq 3\bar{f}T^*+\hat{\alpha}K\bar{f}+\hat{\alpha}(\bar{\rho}+\frac{\bar{\beta}}{2}\bar{\rho}_{\max}+\|\bar{x}\|_{\infty})\sqrt{\frac{V_h(\mu,\mu_1)}{2\sigma}}(1+\sqrt{2})\sqrt{T},
\end{split}
\end{equation}
where the second inequality holds by \eqref{eqn:frame_regret} and the third inequality holds due to the fact that $\hat{\alpha}\geq \alpha_i$ for any $i\in[3,K]$.

Since $K=\left \lceil \log_2(T/T^*) \right \rceil=O(\log(T))$,  it holds for any sequence $y$ that
\begin{equation}
\lim_{T\rightarrow \infty}\frac{1}{T}\sum_{t=1}^{T} f_t(x_t^*)-\hat{\alpha} f_t(x_t) \leq  0, 
\end{equation}
indicating an asymptotic competitive ratio of $CR^{\experttwo}=\frac{1}{\hat{\alpha}}=\min_{m\in[M]}\frac{\rho_m+\Delta\rho_m}{\bar{x}_m}$.
\end{proof}

\section{Proof of Theorem~\ref{thm:robustness}}\label{appendix:proof_robustness}
\begin{proof}
To prove the wost-case robustness of \ouralg, we need to prove that there exists at least one feasible action in each round.
We prove by induction that  $\check{x}_t=\min\{x_t^{\dagger},B_t+E_t\}$ is always feasible for constraint \eqref{eqn:constraint_new_reservation}. 

When $t=1$, $x_t^{\dagger}$ is obviously a feasible solution of \eqref{eqn:constraint_new_reservation}.  Let $F_t=\sum_{\tau=1}^t f_\tau(x_\tau)$ for any $t\in[T]$. Assume that at round $t-1$, $F_{t-1}-\Delta(x_{t-1})+R\geq \lambda F^{\dagger}_{t-1}$. At round $t$, we have
\begin{equation}\label{eqn:robustnessproof1}
\begin{split}
&F_{t}-\Delta(x_t)+R- \lambda F^{\dagger}_{t}\\
=&F_{t-1}- \lambda F^{\dagger}_{t-1}-\Delta(x_t)+R+f_t(x_t)-\lambda f_t(x^{\dagger}_t)\\
\geq &\left(\Delta(x_{t-1})-\Delta(x_{t})\right)+f_t(x_t)-\lambda f_t(x^{\dagger}_t)\\
=&\lambda L\left(\sum_{m=1}^M |B_{m,t}^{\dagger}-B_{m,t}|^+- |B_{m,t+1}^{\dagger}-B_{m,t+1}|^+\right)+f_t(x_{t})-\lambda f_t(x^{\dagger}_{t}),
\end{split}
\end{equation}
where $B_{m,t+1}=B_{m,t}+{E}_{m,t} - x_{m,t}$ and $B_{m,t+1}^{\dagger}=B_{m,t}^{\dagger}+{E}_{m,t}^{\dagger} - x_{m,t}^{\dagger}$ by the budget dynamics.

Next, we prove $x_t=\check{x}_t$ is always a feasible solution for constraint \eqref{eqn:constraint_new_reservation}. If $x_t=\check{x}_t$, we have $B_{t+1}=B_{t}+E_{t}-\check{x}_t$.
If $B_{m,t}+E_{m,t}\geq x^{\dagger}_{m,t}$ holds for $m$, then $\check{x}_{m,t}=x^{\dagger}_{m,t}$ and we have
\begin{equation}
\begin{split}
|B_{m,t+1}^{\dagger}-B_{m,t+1}|^+&=|B_{m,t}^{\dagger}+E_{m,t}^{\dagger}-B_{m,t}-E_{m,t}|^+\\
&=|\min\{B_{m,t}^{\dagger}+\hat{E}_{m,t},B_{\max}\}-\min\{B_{m,t}+\hat{E}_{m,t},B_{\max}\}|^+\\
&\leq |B_{m,t}^{\dagger}-B_{m,t}|^+,
\end{split}
\end{equation}
where the last inequality holds by $1$-Lipschitz of the function $\min\{\cdot, B_{\max}\}$.  On the other hand, if $B_{m,t}+E_{m,t}< x^{\dagger}_{m,t}$ holds for $m$, then $\check{x}_{m,t}=B_{m,t}+E_t$ holds for $m$. Thus 
\begin{equation}
\begin{split}
&|B_{m,t}^{\dagger}-B_{m,t}|^+- |B_{m,t+1}^{\dagger}-B_{m,t+1}|^+\\
=&(B_{m,t}^{\dagger}-B_{m,t})-|B_{m,t}^{\dagger}+E_{m,t}^{\dagger}-x_{m,t}^{\dagger}-B_{m,t+1}|^+\\
=&-B_{m,t}-E_{m,t}^{\dagger}+x_{m,t}^{\dagger}\\
\geq&x_{m,t}^{\dagger}-\check{x}_{m,t},
\end{split}
\end{equation}
where the first equality holds because $\min\{B_{m,t}+\hat{E}_{m,t},B_{\max}\}=B_{m,t}+E_{m,t}< x^{\dagger}_{m,t}\leq B^{\dagger}_{m,t}+E_{m,t}^{
\dagger}=\min\{B_{m,t}^{\dagger}+\hat{E}_{m,t},B_{\max}\}
$, so $B_{m,t}\leq B_{m,t}^{\dagger}$, the second equality holds because $B_{m,t+1}=B_{m,t}+E_{m,t}-\check{x}_{m,t}=0$, and the inequality holds because $E_{m,t}\geq E_{m,t}^{\dagger}$ given $B_{m,t}\leq B_{m,t}^{\dagger}$.
Thus we have for any $m\in[M]$,
\begin{equation}
L\left(\sum_{m=1}^M |B_{m,t}^{\dagger}-B_{m,t}|^+- |B_{m,t+1}^{\dagger}-B_{m,t+1}|^+\right)\geq L(x_{m,t}^{\dagger}-\bar{x}_{m,t}).
\end{equation}

Thus,
by the Lipschiz continuity of $f$, we have
\begin{equation}
\begin{split}
&f_t(x_{t}^{\dagger})-f_t(\check{x}_{t})\\
\leq& \sum_{m=1}^ML |x_{m,t}^{\dagger}-\check{x}_{m,t}|\\
\leq& \sum_{m=1}^ML(|B_{m,t}^{\dagger}-B_{m,t}|^+- |B_{m,t+1}^{\dagger}-B_{m,t+1}|^+).
\end{split}
\end{equation}
Continuing with \eqref{eqn:robustnessproof1}, when $x_t=\check{x}_t$, since $\lambda\in[0,1]$, we have
\begin{equation}
F_{t}-\Delta(x_t)+R- \lambda F^{\dagger}_{t}
\geq  (1-\lambda) f_t(\check{x}_{t})\geq 0.
\end{equation}

Thus we prove that there always exists $\check{x}_t=\min\{x_t^{\dagger},B_t+E_t\}$ such that $F_{t}-\Delta(x_t)+R\leq \lambda F^{\dagger}_{t}$ holds for each round $t$. Since $\Delta(x_t)\geq 0$, if \eqref{eqn:constraint_new_reservation} holds for each round, we have \eqref{eqn:constraint_new_reservation} holds for the last round, thus satisfying the worst-case utility constraint \eqref{eqn:constraint_1_learning}.
\end{proof}
\section{Proof of Theorem \ref{thm:average}}\label{appendix:proof_consistency}
\begin{proof}
The ML policy optimally trained aware of the projection for worst-case utility
constraint is the policy that optimizes the average utility that satisfies \eqref{eqn:constraint_new_reservation} for each round. Thus we bound the average utility by bounding the average utility of the policy $\pi^{\circ}$ based on the optimal unconstrained ML policy $\tilde{\pi}^*$ and \expert $\pi^{\dagger}$, i.e. $\pi^{\circ}=\gamma\tilde{\pi}^{*}+(1-\gamma)\pi^{\dagger}$. 
The constructed policy $\pi^{\circ}$ gives the action $x_t^{\circ}=\gamma \tilde{x}_t^{*}+(1-\gamma) x_t^{\dagger}$ where $\tilde{x}_T^*$ is the output of ML policy $\tilde{\pi}^*$ and $x_t^{\dagger}$ is the output of  $\pi^{\dagger}$. 

We first prove that $x_t^{\circ}$ is always a feasible action for the budget constraints. To show this, we prove by induction that the remaining budget $B_t^{\circ}$ of $\pi^{\circ}$ at each round is no less than a linear combination of the remaining budget $\tilde{B}^*$ of $\tilde{\pi}^{*}$ and the remaining budget $B^{\dagger}$ of $\pi^{\dagger}$. 
At the first round, it holds that
\begin{equation}
\begin{split}
B_2^{\circ}&=\min\{ B_1+\hat{E}_1, B_{\max}\} -x_1^{\circ}\\
&=\gamma \tilde{B}_2^*+(1-\gamma) B_2^{\dagger}.
\end{split}
\end{equation}
Assume for the round $t, t>2$, we have
$B_{t}^{\circ}\geq \gamma \tilde{B}_t^*+(1-\gamma) B_t^{\dagger}$. Then we have
\begin{equation}\label{eqn:budgetcombination}
\begin{split}
B_{t+1}^{\circ}&=\min\{ B^{\circ}_t+\hat{E}_t, B_{\max}\} -x_t^{\circ}\\
&\geq \min\{ \gamma \tilde{B}_t^*+(1-\gamma) B_t^{\dagger}+\hat{E}_t, B_{\max}\} - \gamma \tilde{x}^*_t-(1-\gamma) x_t^{\dagger}\\
&\geq \gamma\left(\min\{\tilde{B}_t^* +\hat{E}_t, B_{\max}\}-\tilde{x}^*_t\right)+(1-\gamma)\left(\min\{B_t^{\dagger}+\hat{E}_t, B_{\max}\}-x_t^{\dagger}\right)\\
&=\gamma \tilde{B}_{t+1}^*+(1-\gamma) B_{t+1}^{\dagger},
\end{split}
\end{equation}
where the second inequality holds because $\min\{\cdot, B_{\max}\}$ is a concave function. Thus, for any round $t\in[T]$, we have $B_{t}^{\circ}\geq \gamma \tilde{B}_{t}^*+(1-\gamma) B_{t}^{\dagger}$. Since the ML policy and \expert both guarantee that $\tilde{B}_t^*\geq 0$ and $B_t^{\dagger}\geq 0$, we have  $B_{t}^{\circ}\geq 0$ which means $x_t^{\circ}$ is a feasible action for budget constraints.

Next, we need to find an $\gamma$ such that the policy $\pi^{\circ}$ satisfy the robustness constraints.
By the robust algorithm design, we need to satisfy the robust constraint for each step $t$ which can be expressed as
\begin{equation}\label{eqn:rboustconstraint_circ}
\sum_{i=1}^t f_i(x_i^{\circ})\geq \lambda \sum_{i=1}^tf_i(x_i^{\dagger})+\lambda L\sum_{m=1}^M |B_{m,t+1}^{\dagger}-B^{\circ}_{m,t+1}|^+-R
\end{equation}
By Lipschitz continuity of $f_t$, we have $f_i(x_i^{\dagger})\leq f_i(x_i^{\circ})+L\|x_i^{\dagger}-x_i^{\circ}\|_1$ (We can use $L^1$-norm since it returns the largest value among $L^p$-norms ($p\geq 1$). ), and thus get a sufficient condition for the robust constriant \eqref{eqn:rboustconstraint_circ} as
\begin{equation}
-\lambda L\sum_{i=1}^t\|x_i^{\circ}-x_i^{\dagger}\|_1-\lambda L\sum_{m=1}^M |B_{m,t+1}^{\dagger}-B^{\circ}_{m,t+1}|^+   \geq (\lambda-1) \sum_{i=1}^t f_i(x_i^{\circ})-R.
\end{equation}
By \eqref{eqn:budgetcombination} and the monotonicity of ReLU operation, we have \begin{equation}\label{eqn:averageproof1}
|B_{m,t+1}^{\dagger}-B^{\circ}_{m,t+1}|^+\leq |B_{m,t+1}^{\dagger}-\gamma\tilde{B}_{m,t+1}^*-(1-\gamma) B_{m,t+1}^{\dagger}|^+=\gamma |B_{m,t+1}^{\dagger}- \tilde{B}_{m,t+1}^*|^+.
\end{equation}
Substituting the expressions of $x_t^{\circ}$ and \eqref{eqn:averageproof1} into the inequality, the sufficient condition for the robust constraint \eqref{eqn:rboustconstraint_circ} becomes
\begin{equation}
-\gamma\lambda L\sum_{i=1}^t\|\tilde{x}_i^{*}-x_i^{\dagger}\|_1-\gamma\lambda L\sum_{m=1}^M |B_{m,t+1}^{\dagger}- \tilde{B}_{m,t+1}^*|^+    \geq (\lambda-1) \sum_{i=1}^t f_i(x_i^{\circ})-R.
\end{equation}

By the definition of $B_t^{\dagger}$ and $\tilde{B}_t^*$, we have 
\begin{equation}
\begin{split}
&\sum_{m=1}^M|B_{m,t+1}^{\dagger}-\tilde{B}_{m,t+1}^*|^+\\=&\sum_{m=1}^M\left|\left(\min\{B_{m,t}^{\dagger} +\hat{E}_{m,t}, B_{\max}\}-x_{m,t}^{\dagger}\right)-\min\{\tilde{B}_{m,t}^*+\hat{E}_{m,t}, B_{\max}\}-\tilde{x}^*_{m,t}\right|^+\\
\leq &\sum_{m=1}^M|B_{m,t}^{\dagger}-\tilde{B}_{m,t}^*|^++|x_{m,t}^{\dagger}-\tilde{x}^*_{m,t}|^+
\leq  \sum_{i=1}^t\sum_{m=1}^M|x_{m,t}^{\dagger}-\tilde{x}^*_{m,t}|^+\leq \sum_{i=1}^t\|x_t^{\dagger}-\tilde{x}^*_t\|_1,
\end{split}
\end{equation}
where the second inequality holds by 1-Lipschitz of $\min\{\cdot, B_{\max}\}$, and the second inequality holds by iteratively applying the first inequality. Thus, the sufficient condition for the robust constraint \eqref{eqn:rboustconstraint_circ} becomes
\begin{equation}
2\gamma\lambda L\sum_{i=1}^t\|\tilde{x}_i^{*}-x_i^{\dagger}\|_1 \leq (1-\lambda) \sum_{i=1}^t f_i(x_i^{\circ})+R.
\end{equation}

Since $(1-\lambda) \sum_{i=1}^t f_i(x_i^{\circ})\geq 0$, if $\gamma\in[0,1]$ satisfies
\begin{equation}
\gamma\leq \min\left\{1, \frac{R}{2\lambda L\sum_{i=1}^t\|\tilde{x}_i^*-x_i^{\dagger}\|_{1}}\right\},
\end{equation}
then $x_t^{\circ}$ satisfies the robust constraint \eqref{eqn:rboustconstraint_circ}  for each round $t$.

Thus, by the definition of $\theta=\max_{y}\sum_{t=1}^T\|\tilde{x}_t^*-x_t^{\dagger}\|_1$, we further have the sufficient condition that $\hat{x}_t$ satisfies the robust constraint is
$\gamma\in[0,1]$ satisfies
\begin{equation}
\gamma\leq \min\left\{1, \frac{R}{2\lambda L\theta}\right\}:=\gamma_{\lambda, R},
\end{equation}

Next, we can bound the average utility of $\pi^{\circ}=\gamma_{\lambda, R}\tilde{\pi}^*+(1-\gamma_{\lambda, R})\pi^{\dagger}$ which is also the bound of the average utility of the proposed policy.
Since the function $f$ is $L-$Lipschitz continuous, then we have
\begin{equation}
\begin{split}
\mathbb{E}_{y}\left[F_T^{\pi^{\circ}}(y)\right]&=\mathbb{E}_{y}\left[F_T^{\tilde{\pi}^*}(y)\right]-\mathbb{E}_{y}\left[\left|F_T^{\pi^{\circ}}(y)- F_T^{\tilde{\pi}^*}(y)\right|\right]\\
&\geq \mathbb{E}_{y}\left[F_T^{\tilde{\pi}^*}(y)\right]-L\mathbb{E}_{y}\left[ \sum_{t=1}^T\|x_t^{\circ}-\tilde{x}_t^*\|\right]\\
&=\mathbb{E}_{y}\left[F_T^{\tilde{\pi}^*}(y)\right]-L(1-\gamma_{\lambda,R})\mathbb{E}_{y}\left[\sum_{t=1}^T \|x_t^{\dagger}-\tilde{x}_t^{*}\|\right]\\
&=\mathbb{E}_{y}\left[F_T^{\tilde{\pi}^*}(y)\right]-L\max\{0,1-\frac{R}{2\lambda L\theta}\}\mathbb{E}_{y}\left[\sum_{t=1}^T \|x_t^{\dagger}-\tilde{x}_t^{*}\|\right],
\end{split}
\end{equation}
where the inequality holds by the Lipschitz continuity of reward functions, and the second equality holds since $x_t^{\circ}-\tilde{x}_t^*=(1-\gamma_{\lambda, R})(x_t^{\dagger}-\tilde{x}_t^*)$.
Since the ML policy $\tilde{\pi}^{\circ}$ is optimally trained under the constraint \eqref{eqn:rboustconstraint_circ}, we have  $\mathbb{E}_{y}\left[F_T^{\ouralg(\tilde{\pi}^{\circ})}(y)\right]\geq \mathbb{E}_{y}\left[F_T^{\pi^{\circ}}(y)\right]$, so we prove the average bound in our theorem.
\end{proof}

\end{document}